\documentclass[11pt,a4paper]{article}

\usepackage{preamble}

\begin{document}
\setlength{\abovedisplayskip}{6pt}
\setlength{\belowdisplayskip}{6pt}

\title{On estimating operator norm distance, with optimal trace distance estimation when one state is pure}

\author[1]{Yupan Liu\thanks{Email: \url{yupan.liu@epfl.ch}}}
\author[2]{Qisheng Wang\thanks{Email: \url{QishengWang1994@gmail.com}}}
\author[3]{Zhan Yu\thanks{Email: \url{yu.zhan@u.nus.edu}}}
\affil[1]{School of Computer and Communication Sciences, \'Ecole Polytechnique F\'ed\'erale de Lausanne}
\affil[2]{School of Computer Science, Shanghai Jiao Tong University}
\affil[3]{Centre for Quantum Technologies, National University of Singapore}
\date{}

\maketitle
\pagenumbering{roman}
\thispagestyle{empty}

\begin{abstract}
    We investigate the computational complexity of estimating the operator norm distance ${\rm T}_{\infty}(\rho_0,\rho_1)$, defined via the operator norm $\|A\|_{\infty} \coloneqq \sigma_{\max}(A)$, where $\sigma_{\max}(A)$ is the largest singular value of $A$, given ${\rm poly}(n)$-size state-preparation circuits of $n$-qubit quantum states $\rho_0$ and $\rho_1$. We provide efficient quantum estimators for the operator norm distance whose complexity is \emph{independent} of the rank (and thus the dimension) of the states:
    \begin{enumerate}[label={\upshape(\arabic*)}]
        \item When one state is pure, we establish an optimal quantum estimator using $\Theta(1/\epsilon)$ queries to the state-preparation circuits. Consequently, for constant additive error, say $\epsilon=1/5$, our estimator runs in ${\rm poly}(n)$ time. Since the operator norm distance ${\rm T}_{\infty}(|\psi\rangle\!\langle\psi|,\rho)$ is \emph{exactly half} of the trace distance ${\rm T}(|\psi\rangle\!\langle\psi|,\rho)$, our result gives a rank-independent query complexity for estimating ${\rm T}_{\infty}(|\psi\rangle\!\langle\psi|,\rho)$ and ${\rm T}(|\psi\rangle\!\langle\psi|,\rho)$, whereas the approaches due to \hyperlink{cite.vACGN23}{van Apeldoorn, Cornelissen, Gily{\'{e}}n, and Nannicini (SODA 2023)} and \hyperlink{cite.WZ24}{Wang and Zhang (TIT 2024)} have query complexity scaling at least linearly with ${\rm rank}(\rho)$, which can be $\exp(n)$ in general.
        In addition, our query complexity matches the optimal bound when \emph{both} states are pure by \hyperlink{cite.Wang24pureQSD}{Wang (TIT 2024)}.
        
        \item  For general quantum states, we also provide a quantum estimator using $\widetilde{O}(1/\epsilon^{3/2})$ queries to the state-preparation circuits, which shows that the corresponding promise problem is ${\sf BQP}$-complete and improves the ${\sf QMA}$ upper bound sketched by \hyperlink{cite.LW25Lalpha}{Liu and Wang (ESA 2025)}. Together with an $\Omega(1/\epsilon)$ quantum query complexity lower bound, this leaves only square-root room for improvement. 
    \end{enumerate}
    The key intuition behind our estimators is that, when one state is pure, the pure state $|\psi\rangle$ has overlap at least $1/2$ with the top unit eigenvector of $|\psi\rangle\!\langle\psi|-\rho$, reflecting a structural feature specific to the operator norm distance. 
\end{abstract}

\newpage
\tableofcontents
\thispagestyle{empty}

\newpage
\pagenumbering{arabic}
\section{Introduction}

Testing the closeness between quantum states is a central topic in quantum property testing~\cite{MdW16}, which studies efficient quantum algorithms for testing properties of quantum objects and extends classical distribution testing (see e.g.,~\cite{Canonne20}) to the non-commutative setting. This problem is also closely connected to the verification of quantum devices, such as $Q_0$ and $Q_1$, which prepare the corresponding quantum states $\rho_0$ and $\rho_1$. In (tolerant) quantum state testing, one seeks to design efficient quantum algorithms that determine whether $\rho_0$ is $\sfrac{2}{5}$-far from or $\sfrac{1}{5}$-close to $\rho_1$ with respect to a chosen closeness measure, typically the trace distance, which is induced by the $\ell_1$ norm. 

Importantly, the trace distance is one of the most widely used measures of closeness because it plays a central role in several research areas: quantum hypothesis testing and state discrimination~\cite{Hel67,Hol73}, information-theoretically secure quantum cryptography~\cite{BHLMO05,RK05,PR22}, classical and statistical zero-knowledge~\cite{SV97,Wat02,Wat09}, and certification and benchmarking~\cite{GLN05,BOW19,EHW+20}. 
Beyond the $\ell_1$ norm distance, namely the total variation distance and the trace distance, (tolerant) closeness testing between probability distributions or quantum states with respect to the $\ell_\alpha$ norm distance, for instance, 
\[\Talpha(\rho_0,\rho_1) \coloneqq \frac{1}{2} \norm*{\rho_0-\rho_1}_\alpha, \quad\text{where }\norm{A}_{\alpha} \coloneqq \tr(\abs{A}^{\alpha})^{1/\alpha},\] 
has also been studied in~\cite{Waggoner15,LW25Lalpha}. Interestingly, for all real $\alpha \geq 1$ for which $\Talpha(\rho_0,\rho_1)$ is a distance metric, a dichotomy emerges: 
\begin{itemize}
    \item In the $\ell_1$ norm case, the sample complexity, namely the number of samples drawn from the distributions, scales polynomially with the support size of those distributions~\cite{VV17,CDVV14,JHW18}. A similar phenomenon holds for quantum states: the number of state copies required scales polynomially with the dimension or rank of the states~\cite{BOW19}. 
    \item For all real $\alpha>1$, the sample complexity becomes \emph{independent} of the support size~\cite{Waggoner15}. Analogously, rank-independent query and sample complexity upper bounds for estimating the quantum $\ell_\alpha$ distance, defined via the Schatten norm, were also obtained in~\cite{LW25Lalpha}. 
\end{itemize}

However, the case $\alpha=\infty$ remains mysterious. 
On the classical side, the support-size-independent (and thus dimension-independent) sample complexity upper bounds in~\cite{Waggoner15} already include this endpoint case. 
On the quantum side, however, the operator-norm-distance case is still listed as an open problem in~\cite{LW25Lalpha}: 

\begin{quote}
    What are the query and sample complexity upper bounds for estimating the operator norm distance $\Tinfty(\rho_0,\rho_1)$? Can one obtain such bounds in a \emph{rank-independent} manner, similar to the classical setting, or is rank dependence \emph{unavoidable}? 
\end{quote}

Why is it important to study the operator norm distance? While the trace distance ($\alpha=1$) serves as the most important case in the quantum $\ell_\alpha$ distance hierarchy, defined via the Schatten norm, the operator norm distance ($\alpha=\infty$), namely
\[ \Tinfty(\rho_0,\rho_1) \coloneqq \frac{1}{2}\norm*{\rho_0-\rho_1}_{\infty}, \]
is the most relaxed notion of closeness precisely because it is the endpoint of the hierarchy, and thus may be viewed as another particularly important case. The operator norm distance captures the largest spectral deviation between two states, and therefore arises as a closeness measure in quantum state tomography~\cite{CKW+16,XK16,vACGN23}, as well as an intermediate analytic tool that can sometimes be converted into trace-distance guarantees in~\cite{GKKT20}. This spectral viewpoint is also closely connected to compressed-sensing and low-rank recovery approaches, which allow high-dimensional states to be reconstructed from far fewer measurements~\cite{Liu11,Gross11}. Beyond measuring the closeness between states, related operator-norm notions also play an important role in quantum process and unitary tomography~\cite{HKOT23} and in quantifying the precision of Hamiltonian simulation~\cite{BCK15,HHKL23}. 

In this work, we provide an affirmative answer to these questions on estimating $\Tinfty(\rho_0,\rho_1)$. 
We next present our main results and discuss their significance.

\subsection{Main results}

We start by stating our first main result, where $\td(\rho_0,\rho_1)\coloneqq \td_{\alpha=1}(\rho_0,\rho_1)=\frac{1}{2}\tr\abs*{\rho_0-\rho_1}$:

\begin{theorem}[Optimal quantum estimator for the trace and operator norm distances when one state is pure, informal]
    \label{thm-informal:OnePureQSD-query}
    Given quantum query access to the state-preparation circuits of two $n$-qubit quantum states, denoted by $\rho$ and $\ketbra{\psi}{\psi}$, without prior knowledge of which input circuit prepares the pure state, the quantum query complexity for estimating $\td(\rho,\ketbra{\psi}{\psi})$ or $\Tinfty(\rho,\ketbra{\psi}{\psi})$ to within additive error $\epsilon$ is $\Theta(1/\epsilon)$. 
\end{theorem}

Here, the lower bounds in \Cref{thm-informal:OnePureQSD-query} follow directly from the corresponding quantum query lower bounds for estimating $\td(\ketbra{\psi_0}{\psi_0},\ketbra{\psi_1}{\psi_1})$ and $\Tinfty(\ketbra{\psi_0}{\psi_0},\ketbra{\psi_1}{\psi_1})$ to within additive error $\epsilon$ in the case where both states are pure, established in~\cite[Theorem V.2]{Wang24pureQSD} and~\cite[Theorem 22(1)]{LW25Lalpha}, respectively.

\begin{theorem}[Quantum estimator for the operator norm distance, informal]
    \label{thm-informal:QSDinfty-query}
    Given quantum query access to the state-preparation circuits of $n$-qubit quantum states $\rho_0$ and $\rho_1$, there exists an explicit quantum algorithm that estimates $\Tinfty(\rho_0,\rho_1)$ to within additive error $\epsilon$ with query complexity $\widetilde{O}(1/\epsilon^{3/2})$, where $\widetilde{O}(\cdot)$ suppresses poly-logarithmic factors.
\end{theorem}

It is worth noting that the quantum query complexity lower bound $\Omega\rbra{1/\epsilon}$ in \Cref{thm-informal:OnePureQSD-query} also applies to \Cref{thm-informal:QSDinfty-query}, and thus our upper bound leaves only a \emph{square-root} gap. 
Moreover, the quantum query upper bound in \Cref{thm-informal:QSDinfty-query} is strictly better than the corresponding upper bounds known for every real $\alpha>1$, and is in fact \emph{super-quadratically} better except for the case of $\alpha=2$. 
Indeed, for even integers $\alpha \geq 2$, the standard upper bound is $O(1/\epsilon^\alpha)$, following directly from~\cite{BCWdW01,EAO+02}; whereas for general real $\alpha>1$ (including odd integers $\alpha\geq 3$), the best known upper bound is $O\rbra[\big]{1/\epsilon^{\alpha+1+\frac{1}{\alpha-1}}}$ as shown in~\cite{LW25Lalpha}. 
These comparisons further support the view that the operator norm distance, as the endpoint of the quantum $\ell_\alpha$ distance hierarchy, is its \emph{easiest} case. 

Prior to this work, there are two classes of known approaches that can be used to estimate the operator norm distance given state-preparation circuits:

\begin{itemize}
    \item \textbf{General case:} 
    The operator norm distance between two quantum states can be estimated directly by quantum state tomography in the operator norm distance, whose current best approach is due to \cite[Theorem 45]{vACGN23} with query complexity $\widetilde{O}\rbra{d/\epsilon}$, where $d = 2^n$ is the dimension of the quantum states and $\epsilon$ is the desired precision.
    In comparison, our estimator in \cref{thm-informal:QSDinfty-query} with query complexity $\widetilde{O}\rbra{1/\epsilon^{3/2}}$ removes the dependence on $d$.
    In particular, when the state-preparation circuits are of size $\poly\rbra{n}$ (which is usually assumed when considering the computational complexity of quantum state testing, e.g., \cite{Wat02,RASW23}) and $\epsilon \geq 1/\poly\rbra{n}$ (e.g., $\epsilon = 0.1$ is a constant), our estimator in \cref{thm-informal:QSDinfty-query} runs in time $\poly\rbra{n}$, whereas the approach based on \cite{vACGN23} runs in time $\exp\rbra{n}$, therefore achieving an \textit{exponential} speedup. 
    \item \textbf{When one state is pure:} 
    In this case, the operator norm distance is equivalent to the trace distance (up to a constant factor). 
    There are quantum estimators for trace distance given in \cite{WGL+22,WZ24} and the current best query complexity is $\widetilde{O}\rbra{r/\epsilon^2}$ due to \cite{WZ24}, where $r$ is the larger rank of \textit{both} quantum states. 
    When one state, say $\rho_0$, is pure, the trace distance estimator has query complexity $\widetilde{O}\rbra{\rank\rbra{\rho_1}/\epsilon^2}$.
    In comparison, our estimator in \cref{thm-informal:OnePureQSD-query} with query complexity $O\rbra{1/\epsilon}$ removes the dependence on the larger rank. 
    In particular, when the state-preparation circuits are of size $\poly\rbra{n}$ and $\epsilon \geq 1/\poly\rbra{n}$, our estimator in \cref{thm-informal:OnePureQSD-query} runs in time $\poly\rbra{n}$, whereas the approach based on \cite{WZ24} runs in time $\exp\rbra{n}$ (note that $r$ can be as large as $d = 2^n$), therefore achieving an \textit{exponential} speedup. 
    When both states are pure, there are optimal quantum estimators for trace distance \cite{Wang24pureQSD,FW25} with query complexity $\Theta\rbra{1/\epsilon}$ and our estimator in \cref{thm-informal:OnePureQSD-query} reproduces their results. 
\end{itemize}

It should be noted that our results in \cref{thm-informal:OnePureQSD-query,thm-informal:QSDinfty-query} further lead to new quantum sample complexities of estimating the operator norm distance by the quantum sample-to-query lifting method recently developed in \cite{WZ23,WZ24b,TWZ25,CWZ25}. 
Specifically, \cref{thm-informal:OnePureQSD-query} implies a sample complexity of $\Theta\rbra{1/\epsilon^2}$ for estimating the trace distance when one state is pure (see \cref{thm:sample-pure}), which (i) reproduces the result in \cite{WZ24c} when both states are pure and (ii) matches the sample complexity lower bound in \cite[Theorem B.2]{Wang24pureQSD}; 
\cref{thm-informal:QSDinfty-query} implies a sample complexity of $\widetilde{O}\rbra{1/\epsilon^3}$ for estimating $\Tinfty(\rho_0,\rho_1)$ (see \cref{thm:sample-general}), which \textit{exponentially} improves the prior tomography-based approach due to \cite[Theorem V.4]{CHL+23} with sample complexity $O\rbra{d/\epsilon^2}$. 

\vspace{1em}
Let \QSDinfty{} denote the promise problem of testing the closeness between quantum states with respect to the operator norm distance, namely, deciding whether $\Tinfty(\rho_0,\rho_1)$ is at least $2/5$ or at most $1/5$.\footnote{See \Cref{def:QSDalpha} for a formal definition.} We also consider the following restricted variants \OnePureQSDinfty{} and, more generally, $\OnePureQSD_{\alpha}$ for the quantum $\ell_\alpha$ distance, in which at least one of $\rho_0$ and $\rho_1$ is pure; and \PureQSDalpha{}, in which both states are pure.

When the state-preparation circuits have polynomial-size descriptions,  the quantum algorithms in \Cref{thm-informal:OnePureQSD-query,thm-informal:QSDinfty-query} run in time polynomial in $n$. Therefore, in the white-box input model, these query complexity upper bounds imply polynomial-time quantum algorithms, i.e., \BQP{} upper bounds, for the corresponding promise problems, improving the prior \QMA{} upper bound sketched in \cite[Footnote 8]{LW25Lalpha}. 
Combining these upper bounds with the \BQP-hardness results established for $\PureQSD_\alpha$ when $1\leq \alpha \leq \infty$ in~\cite[Theorem 21]{LW25Lalpha}, and using the fact that the trace distance and the operator norm distance correspond to the two endpoints of the quantum $\ell_\alpha$ distance hierarchy, we obtain the following corollaries:
\begin{corollary}
    \QSDinfty{} is \BQP{}-complete. 
\end{corollary}

\begin{corollary}
    \label{corr:OnePureQSDalpha-BQPcomplete}
    For every $\alpha\in[1,\infty]$, $\OnePureQSD_\alpha$ is \BQP{}-complete.
\end{corollary}

\subsection{Proof techniques}

We begin by explaining the intuition behind our quantum estimators for the operator norm distance. One natural approach is based on phase estimation and proceeds as follows: 
\begin{enumerate}[label={\upshape(\arabic*)}]
    \item Apply an appropriate version of phase estimation,\footnote{For example, an approach with query complexity $\widetilde O\rbra{\sqrt{d}/\epsilon}$ is known in \cite[Lemma 50]{vAGGdW20} based on phase estimation, equipped with appropriate versions of Hamiltonian simulation and quantum minimum finding, where $d$ is the dimension of $\Delta$.} constructed based on a unitary dilation of $\Delta \coloneqq (\rho_0-\rho_1)/2$,\footnote{Given the state-preparation circuits of $\rho$ and $\ketbra{\psi}{\psi}$, or oracle access to these circuits, one can efficiently construct a unitary dilation of $\Delta$ via the techniques of Linear Combination of Unitaries (LCU)~\cite{CW12,BCC+15} and purified density matrix \cite{LC19}. See \Cref{lemma:dilation-Delta} for an explicit construction.} to a state supported on the \emph{top eigenspace} of $\Delta$, which corresponds to the largest eigenvalue $\lambda_{\max}(\Delta)$. 
    \item Repeat the same construction for $-\Delta$, and then take the maximum of the two resulting estimates, thereby obtaining an estimate of $\sigma_{\max}(\Delta) = \max\cbra{\lambda_{\max}(\Delta),\lambda_{\max}(-\Delta)}$.
\end{enumerate}
However, in general, preparing such a state in the top eigenspace appears to require additional witness information. Consequently, this approach seems to yield only a \QMA{} upper bound, as sketched in~\cite[Footnote 8]{LW25Lalpha}. To turn this idea into an actual efficient algorithm, we must therefore resolve two issues: (i) find a quantum state with \emph{sufficiently large} overlap with the top eigenspace of $\Delta$; and (ii) identify the appropriate version of phase estimation. 

\paragraph{The input states provide a warm start when one state is pure.}
We first address the former issue in the setting where one state is pure, and our construction even extends naturally to estimating the trace distance. Previous optimal quantum query and sample algorithms for estimating the trace distance between two pure states $\ket{\psi_0}$ and $\ket{\psi_1}$, developed in~\cite{WZ24c,Wang24pureQSD} (and implicitly in~\cite{FW25}), rely crucially on the following identity, which is implicit in~\cite{FvdG99} and can be viewed as capturing the \emph{two-dimensional geometry of pure states}:
\begin{equation}
    \label{eq:pure-state-traceDist-fidelity}
    \td(\ketbra{\psi_0}{\psi_0},\ketbra{\psi_1}{\psi_1})^2 + \F(\ketbra{\psi_0}{\psi_0},\ketbra{\psi_1}{\psi_1})^2 = 1.
\end{equation}

Since the square of the fidelity, $\F(\ketbra{\psi_0}{\psi_0},\ketbra{\psi_1}{\psi_1})^2$, between pure states is simply the overlap $\abs{\innerprod{\psi_0}{\psi_1}}^2$ and the largest singular value of $\ketbra{\psi_0}{\psi_0}-\ketbra{\psi_1}{\psi_1}$ is 
\[ \sigma_{\max}(\ketbra{\psi_0}{\psi_0}-\ketbra{\psi_1}{\psi_1}) = \sqrt{1-\abs{\innerprod{\psi_0}{\psi_1}}^2}, \]
we obtain the following identity relating the trace distance and the operator norm distance: 
\begin{equation}
    \label{eq:pure-state-traceDist-alternative}
    \td(\ketbra{\psi_0}{\psi_0},\ketbra{\psi_1}{\psi_1}) = \sigma_{\max}(\ketbra{\psi_0}{\psi_0}-\ketbra{\psi_1}{\psi_1}) = 2 \cdot \Tinfty(\ketbra{\psi_0}{\psi_0},\ketbra{\psi_1}{\psi_1}).
\end{equation}

While \Cref{eq:pure-state-traceDist-fidelity} is specific to the case in which both states are \emph{pure}, the coincidence between the trace distance and the operator norm distance in \Cref{eq:pure-state-traceDist-alternative} extends to the more general setting in which only one state is pure, as proven in~\cite[Theorem 18(2)]{LW25}:
\begin{equation}
    \label{eq:one-pure-state-traceDist}
    \td(\rho,\ketbra{\psi}{\psi}) = 2\cdot \Tinfty(\rho,\ketbra{\psi}{\psi}).
\end{equation}

Therefore, in the single-pure-state setting, estimating the trace distance is \emph{equivalent} to estimating the operator norm distance. We now show that there exists a quantum state with sufficiently large overlap with the top eigenspace of $\Delta$, as established by our \emph{single-pure-state warm-start lemma} (\Cref{lemma:one-pure-state-warm-start}), which shows that $\ket{\psi}$ itself already provides a \emph{warm start}. Indeed, assuming without loss of generality that $\rho_0=\ketbra{\psi}{\psi}$ and $\rho_1=\rho$, the state $\ket{\psi}$ has squared overlap at least $1/2$ with the eigenvector $\ket{v_+}$ corresponding to the \emph{unique} positive eigenvalue $\lambda_{\max}(\Delta)$, and therefore the largest eigenvalue, of $\Delta$:
\begin{equation}
    \label{eq:warm-start-lemma}
    \abs*{\innerprod{v_+}{\psi}}^2 \geq \frac{1}{2}+\norm{\Delta}_\infty > \frac{1}{2}. 
\end{equation}

Combining this warm start with an appropriate version of phase estimation, we obtain the $O(1/\epsilon)$ quantum query complexity upper bound, matching the pure-vs.-pure case in~\cite{Wang24pureQSD}. 
The underlying reason is essentially a \emph{rank-one perturbation} phenomenon: since $\rho$ is positive semidefinite, the restriction of $\Delta$ to the subspace orthogonal to $\ket{\psi}$ is negative semidefinite. Consequently, any eigenvector corresponding to a positive eigenvalue must have substantial overlap with $\ket{\psi}$, and in particular $\Delta$ has \emph{exactly one} positive eigenvalue when $\Delta \neq 0$.
This fact also explains the intuition behind \Cref{eq:pure-state-traceDist-alternative,eq:one-pure-state-traceDist}. 

Furthermore, the proof of this warm-start lemma relies crucially on the fact that $\rho^2 \preceq \rho$, which becomes an equality when $\rho$ is pure. Therefore, the bound in \Cref{eq:warm-start-lemma} is tight, and is already saturated when both the states are pure. In this sense, this rank-one perturbation phenomenon can be viewed as a natural counterpart to \Cref{eq:pure-state-traceDist-fidelity} when the assumption that both states are pure is relaxed to the case where only one state is pure. 

\paragraph{From na\"ive phase estimation to maximum phase estimation.}
To resolve the latter issue, one might first consider the na\"ive approach of combining Kitaev's standard phase estimation~\cite{Kitaev95} with Hamiltonian simulation, such as that in~\cite{BCK15}: namely, one simulates $e^{it\Delta}$ and recovers an extremal eigenvalue of $\Delta$ from the corresponding eigenphase. While this approach is conceptually natural, it is not well suited to our setting. Indeed, estimating an eigenvalue to within additive error $\epsilon$ in this way requires resolving phases to precision $O(\epsilon)$, and thus using simulation time $t=\Theta(1/\epsilon)$. Implementing such a long-time evolution then incurs an additional factor proportional to $t$, yielding only a $\widetilde{O}(1/\epsilon^2)$-type quantum query upper bound, which is quadratically worse than the optimal bound in \Cref{thm-informal:OnePureQSD-query}. 

The right choice, instead, is to work with a unitary dilation of $\Delta$ and then apply maximum phase estimation~\cite{MdW23} (see also \Cref{lem:max_phase_estimation}) after qubitization~\cite{LC19} (see \Cref{lemma:qubitization} for the specific version we need). Under the resulting monotone correspondence between eigenvalues and eigenphases, the largest eigenvalue of $\Delta$ is converted into the largest eigenphase of the associated unitary. Hence, the problem becomes that of estimating a maximum phase from a state having nontrivial overlap with the top eigenspace. 

Consequently, in the single-pure-state setting, the warm-start lemma (\Cref{lemma:one-pure-state-warm-start}) provides exactly such a state with constant overlap, and this is precisely what enables the maximum-phase-estimation primitive to recover $\lambda_{\max}(\Delta)$, and hence both $\Tinfty(\rho,\ketbra{\psi}{\psi})$ and $\td(\rho,\ketbra{\psi}{\psi})$, within additive error $\epsilon$ using $O(1/\epsilon)$ queries, as stated in \Cref{thm-informal:OnePureQSD-query}. 

\paragraph{The input states provide an eigenvalue-scaled overlap.}
While the same version of phase estimation can also be used to estimate the operator norm distance in the general setting, the key remaining issue is to identify a quantum state with sufficiently large overlap with the top eigenspace of $\Delta = (\rho_0-\rho_1)/2$. 

To this end, rather than using the input states $\rho_0$ and $\rho_1$ directly, we work with their purifications $\ket{\psi_0}$ and $\ket{\psi_1}$. Let $\Pi^{(\Delta)}_{\max}$ denote the projector onto the top eigenspace of $\Delta$, which may have rank greater than $1$ since the top eigenvalue of $\Delta$ need not be \emph{unique}.
Then, the overlap between $\ket{\psi_0}$ and this eigenspace, defined via the Euclidean vector norm $\norm{\cdot}_2$, is
\begin{equation}
    \label{eq:topEigenspace-overlap-def}
    \norm[\big]{\ab\big(I\otimes \Pi^{(\Delta)}_{\max})\ket{\psi_0}}_2^2 = \tr\ab\big(\Pi^{(\Delta)}_{\max}\rho_0).
\end{equation}

It is also worth noting that \Cref{eq:topEigenspace-overlap-def} already captures the warm-start phenomenon in the single-pure-state setting (cf.~\Cref{lemma:one-pure-state-warm-start}). Indeed, when $\rho_0=\ketbra{\psi}{\psi}$ is pure, \Cref{eq:warm-start-lemma} immediately yields the following \emph{constant}-overlap bound:
\begin{equation}
    \label{eq:topEigenspace-overlap-onePure}
    \norm[\big]{\ab\big(I\otimes \Pi^{(\Delta)}_{\max})\ket{\psi_0}}_2^2 = \tr\ab\big(\Pi^{(\Delta)}_{\max}\rho_0)  \geq \frac{1}{2}.
\end{equation} 

When neither $\rho_0$ nor $\rho_1$ is pure, however, this constant-overlap argument from \Cref{eq:topEigenspace-overlap-onePure} no longer applies. Instead, using the projection identity $\Pi^{(\Delta)}_{\max} \Delta \Pi^{(\Delta)}_{\max} = \lambda_{\max}(\Delta) \Pi^{(\Delta)}_{\max}$, together with the overlap formula in \Cref{eq:topEigenspace-overlap-def}, we obtain the weaker but fully general lower bound established in our \emph{top eigenspace overlap lemma} (\Cref{lem:top_eigenspace_overlap}): 
\begin{equation}
    \label{eq:topEigenspace-overlap-bound}
    \norm[\big]{\ab\big(I\otimes \Pi^{(\Delta)}_{\max})\ket{\psi_0}}_2^2 \geq 2 \tr\ab\big(\Pi^{(\Delta)}_{\max}\Delta) = 2 \lambda_{\max}(\Delta) \tr\ab\big(\Pi^{(\Delta)}_{\max}) \geq 2\lambda_{\max}(\Delta).
\end{equation}

Consequently, unlike in the single-pure-state setting, there is in general no constant-overlap warm start. Since the overlap guarantee in \Cref{eq:topEigenspace-overlap-bound} scales only with $2\lambda_{\max}(\Delta)$, the success probability of a single run can be as small as $O(\epsilon)$ at the target precision. To boost the success probability to a constant, say $2/3$, we therefore combine this procedure with amplitude amplification and estimation~\cite{BHMT02} (as in maximum phase estimation~\cite{MdW23}). This subtlety incurs an additional multiplicative factor of $\widetilde{O}(1/\sqrt{\epsilon})$, leading to the desired quantum query upper bound $\widetilde{O}(1/\epsilon^{3/2})$ in \Cref{thm-informal:QSDinfty-query}. 

\subsection{Discussion and open problems}

A natural question is whether our warm-start lemma (see \Cref{lemma:one-pure-state-warm-start}) can be extended to a more general setting. Here, we provide a simple example showing that such a warm start already fails for a pair of rank-two states, thereby revealing an immediate obstruction:
\[ \rho_0^{(\eta)} \coloneqq (1-\eta) \ketbra{0}{0} + \eta\ketbra{1}{1} \quad\text{ and }\quad \rho_1^{(\eta)} \coloneqq (1-\eta)\ketbra{0}{0} + \eta\ketbra{2}{2}. \]

A direct calculation shows that $\Tinfty\rbra[\big]{\rho_0^{(\eta)},\rho_1^{(\eta)}}=\eta/2$ and $\td\rbra[\big]{\rho_0^{(\eta)},\rho_1^{(\eta)}}=\eta$. Moreover, the extremal positive and negative eigenvalues of $\rho_0^{(\eta)}-\rho_1^{(\eta)}$ are $\eta$ and $-\eta$, respectively, with the corresponding eigenvectors $\ket{1}$ and $\ket{2}$. Consequently, the overlaps between these eigenvectors and the input states are 
$\bra{1} \rho_0^{(\eta)} \ket{1} = \eta$ and $\bra{2} \rho_1^{(\eta)} \ket{2} = \eta$.
Since $\eta$ can be chosen arbitrarily small, there is \emph{no constant $c>0$} such that every pair of rank-two states admits an extremal eigenvector of $\rho_0-\rho_1$ with overlap at least $c$ with one of the input states, or even with a linear combination of them. This motivates our first open problem:
\begin{enumerate}[label={\upshape(\roman*)}]
    \item Can one characterize the largest family of mixed states for which one can still prepare a state that has constant overlap with the top eigenspace of $\rho_0-\rho_1$?
\end{enumerate}

A second direction concerns optimality. Our current bounds in \Cref{thm-informal:QSDinfty-query} for the general setting still leave a square-root gap, leading to the following question: 
\begin{enumerate}[label={\upshape(\roman*)}]
    \setcounter{enumi}{1}
    \item Can the quantum query and sample complexity of estimating the operator norm distance be made optimal? More generally, can analogous optimal results be achieved for estimating the quantum $\ell_\alpha$ distance for real $\alpha>1$~\cite{LW25Lalpha}?
\end{enumerate}

\subsection{Related works}
\label{sec:related-works}

Quantum state tomography has been studied with respect to trace distance, fidelity (Bures distance), Frobenius distance (Hilbert--Schmidt distance), and operator norm distance in the literature \cite{HHJ+17,OW16,CHL+23,vACGN23,PSW25,SSW25,PSTW25}.
The closeness testing of quantum states, also known as the quantum state certification, was studied with respect to trace distance, fidelity, $\ell_2$ distance, and $\ell_3$ distance in \cite{BOW19,GL20}. 
The tolerant closeness testing, which is equivalent to closeness estimation, was studied for trace distance, fidelity, and $\ell_\alpha$ distance in \cite{CCC19,WZC+23,GP22,WGL+22,WZ24,LGLW23}. 

\section{Preliminaries}

We assume that the reader is familiar with the basics of quantum computation and quantum information. For background on these subjects, we refer to the standard textbook~\cite{NC10}.

Throughout the paper, we use the following notation. First, for any positive integer $n$, let $\sbra{n} \coloneqq \cbra{1,2,\dots,n}$. Second, $\widetilde{O}(f)$ stands for $O(f\polylog(f))$, whereas $\widetilde{\Omega}(f)$ stands for $\Omega(f/\polylog(f))$. Finally, for a matrix $A$, the Schatten $\alpha$-norm is defined by
\[
\Abs{A}_{\alpha} \coloneqq \rbra[\big]{\tr\rbra{\abs{A}^{\alpha}}}^{1/\alpha} = \rbra*{\tr\rbra[\Big]{\rbra[\big]{A^{\dagger} A}^{\alpha/2}}}^{1/\alpha}. 
\]

\subsection{Closeness measures for quantum states}

We begin by defining the trace distance: 
\begin{definition}[Trace distance]
    Let $\rho_0$ and $\rho_1$ be two quantum states on the same Hilbert space. The trace distance between $\rho_0$ and $\rho_1$ is defined by
    \[\td(\rho_0,\rho_1) \coloneqq \frac{1}{2} \norm{\rho_0-\rho_1}_1 = \frac{1}{2}\tr\rbra{|\rho_0-\rho_1|}.\]
\end{definition}

Next, we define the operator norm distance, a quantity that has also been used in~\cite{GKKT20,vACGN23}, where the corresponding definition is only implicit, while an explicit formal definition is given as a special case of~\cite[Definition 2.3]{LW25Lalpha}:
\begin{definition}[Operator norm distance]
    \label{def:operator-norm-distance}
    Let $\rho_0$ and $\rho_1$ be two quantum states on the same Hilbert space. 
    The operator norm distance between $\rho_0$ and $\rho_1$ is defined by
    \[\Tinfty(\rho_0,\rho_1) \coloneqq \frac{1}{2} \norm*{\rho_0 - \rho_1}_{\infty} = \frac{1}{2} \max\cbra*{ \lambda_{\max}(\rho_0 - \rho_1), \lambda_{\max}(\rho_1 - \rho_0) }.\]
    Here, $\lambda_{\max}(A)$ denotes the largest eigenvalue of the matrix $A$. 
\end{definition}

Notably, the operator norm distance arises as the endpoint case ($\alpha=\infty$) of the quantum $\ell_{\alpha}$ distance~\cite{LW25Lalpha}, which generalizes both the trace distance ($\alpha=1$) and the Hilbert-Schmidt distance ($\alpha=2$) via the Schatten norm. In addition, both the trace distance and the operator norm distance are distance metrics, and hence satisfy the triangle inequality, as stated in~\cite[Lemma 2.4]{LW25Lalpha} (see also~\cite[Proposition 1.16]{AS17}). 

Moreover, the inequalities that relate the operator norm distance to the trace distance were established in~\cite{LW25Lalpha}, as stated in \Cref{lemma:Tinf-vs-T}. Interestingly, when one of the states is pure, the trace distance is \emph{exactly} twice the operator norm distance. 

\begin{lemma}[$\Tinfty$ vs.~$\td$, adapted from~{\cite[Theorem 4.2(2)]{LW25Lalpha}}]
    \label{lemma:Tinf-vs-T}
    Let $\rho_0$ and $\rho_1$ be quantum states on the same Hilbert space. Then, the following bounds hold:  
    \[2 \cdot \Tinfty(\rho_0,\rho_1) \leq \td(\rho_0,\rho_1) \leq 2 \min\cbra*{\rank(\rho_0),\rank(\rho_1)} \cdot \Tinfty(\rho_0,\rho_1).\]
    Furthermore, when $\rho_0=\ketbra{\psi}{\psi}$ and $\rho_1 = \rho$, the inequality becomes an identity: 
    \[2 \cdot \Tinfty(\ketbra{\psi}{\psi},\rho) = \td(\ketbra{\psi}{\psi},\rho). \]
\end{lemma}

\subsection{Closeness testing of quantum states via state-preparation circuits}
\label{subsec:state-closeness-testing}

We first define closeness testing for quantum states with respect to the quantum $\ell_\alpha$ distance, which generalizes the \textsc{Quantum State Distinguishability Problem} (\QSD{}) introduced in~\cite[Section 3.3]{Wat02}, as stated in \Cref{def:QSDalpha}. We denote this promise problem by $\QSDalpha[a,b]$ and also consider the following two restricted variants:   
\begin{definition}[Quantum State Distinguishability Problem with Schatten $\alpha$-norm, \QSDalpha{}, adapted from~{\cite[Definition 17]{LW25Lalpha}}]
	\label{def:QSDalpha}
    Let $Q_0$ and $Q_1$ be quantum circuits acting on $m$ qubits (``input length'') and having $n$ specified output qubits (``output length''), where $m(n)$ is polynomial in $n$. Let $\rho_i$ denote the quantum state obtained by running $Q_i$ on the state $\ket{0}^{\otimes m}$ and tracing out the non-output qubits. Let $a(n)$ and $b(n)$ be efficiently computable functions. The task is to decide whether: 
	\begin{itemize}
		\item \emph{Yes:} The pair of quantum circuits $(Q_0,Q_1)$ satisfies $\Talpha(\rho_0,\rho_1) \geq a(n)$; 
		\item \emph{No:} The pair of quantum circuits $(Q_0,Q_1)$ satisfies $\Talpha(\rho_0,\rho_1) \leq b(n)$.
	\end{itemize}

    \noindent Furthermore, we consider two restricted versions:
    \begin{enumerate}[label={\upshape(\arabic*)}, topsep=0.33em, itemsep=0.33em, parsep=0.33em]
        \item \PureQSD{}\emph{:} Both $\rho_0$ and $\rho_1$ are pure states, and the input length may be larger than the output length. 
        \item \OnePureQSD{}\emph{:} At least one of $\rho_0$ and $\rho_1$ is a pure state, without prior knowledge of which input circuit prepares the pure state. 
    \end{enumerate}
\end{definition}

In this work, we consider the \textit{purified quantum access input model}, as originally introduced in~\cite{Wat02}, in both white-box and black-box settings: 
\begin{itemize}
    \item \textbf{White-box model}: The input of the problem \QSDalpha{} consists of descriptions of quantum circuits $Q_0$ and $Q_1$ that are \emph{ polynomial-size}. In other words, each circuit description consists of a polynomial number of one- and two-qubit gates.
    \item \textbf{Black-box model}: Instead of having explicit circuit descriptions, one is given only query access to $Q_b$, $Q_b^{\dagger}$, controlled-$Q_b$, and controlled-$Q_b^\dagger$, for $b\in\binset$. 
    For convenience, we use ``a query to $Q_b$'' to refer to a query to $Q_b$,  $Q_b^{\dagger}$, controlled-$Q_b$, and controlled-$Q_b^\dagger$.
\end{itemize}

In addition to query complexity, which is defined in the black-box model, \textit{sample complexity} refers to the number of state copies of $\rho_0$ and $\rho_1$ required to test the closeness between these states with respect to the chosen closeness measure. 

\subsection{Unitary dilation}

We will use the notion of unitary dilation, which may be viewed as an \emph{error-free} version of block-encoding with normalization factor $1$~\cite{LC19,GSLW19}.

\begin{definition} [Unitary dilation]
    Let $A$ be an $n$-qubit operator on subsystem $\mathsf{A}$ with $\Abs{A}_{\infty} \leq 1$. 
    An $\rbra{n+m}$-qubit unitary operator $U$ on subsystem $\mathsf{AB}$ is said to be a unitary dilation of $A$ with $m$ ancillary qubits, if $\Pi U \Pi = A_{\mathsf{A}} \otimes \ketbra{0}{0}_{\mathsf{B}}$, where $\Pi = I_{\mathsf{A}} \otimes \ketbra{0}{0}_{\mathsf{B}}$. 
\end{definition}

For our purpose and completeness, we present a simple implementation of a unitary dilation of $\rbra{\rho_0-\rho_1}/2$ for any two quantum states $\rho_0$ and $\rho_1$.

\begin{lemma}[Unitary dilation from state-preparation circuits] \label{lemma:dilation-Delta}
    Given quantum unitary oracles $Q_0$ and $Q_1$ that prepare the purifications of two unknown quantum states $\rho_0$ and $\rho_1$, respectively, we can implement a unitary dilation $U_\Delta$ of $\Delta \coloneqq \rbra{\rho_0-\rho_1}/2$ using $1$ query to each of \mbox{controlled-$Q_0$}, \mbox{controlled-$Q_1$}, and their inverses. 
    Moreover, if each of $Q_0$ and $Q_1$ uses $a$ ancillary qubits, then $U_\Delta$ uses $\rbra{2a+1}$ ancillary qubits. 
\end{lemma}
\begin{proof}
    First, we can implement the unitary dilations $U_0$ and $U_1$ of $\rho_0$ and $\rho_1$, respectively, by the purified density matrix technique \cite[Lemma 7]{LC19}, where $U_b$ uses $2$ queries to $Q_b$ for each $b \in \cbra{0, 1}$. 
    Specifically, suppose that $\rho_0$ and $\rho_1$ are $n$-qubit quantum states. 
    Without loss of generality, we assume that $Q_0$ and $Q_1$ are $\rbra{n+a}$-qubit unitary operators such that 
    \[
    \rbra{Q_b}_{\mathsf{AB}}\ket{0}_{\mathsf{A}}\ket{0}_{\mathsf{B}} = \ket{\psi_b}_{\mathsf{AB}}, \quad \rho_b = \tr_{\mathsf{B}}\rbra{\ketbra{\psi_b}{\psi_b}_{\mathsf{AB}}},
    \]
    where subsystem $\mathsf{A}$ consists of $n$ qubits and subsystem $\mathsf{B}$ consists of $a$ qubits. 
    Let $\mathsf{C}$ be a subsystem consisting of another $a$ qubits. 
    Then, the following $\rbra{n+2a}$-qubit unitary operator
    \[
    U_b = \rbra[\Big]{ \rbra{Q_b^\dag}_{\mathsf{AB}} \otimes I_{\mathsf{C}} } \rbra[\Big]{ I_{\mathsf{A}} \otimes \textup{SWAP}_{\mathsf{BC}} } \rbra[\Big]{ \rbra{Q_b}_{\mathsf{AB}} \otimes I_{\mathsf{C}} }
    \]
    is a unitary dilation of $\rho_b$ with $2a$ ancillary qubits such that $\Pi U_b \Pi = \rbra{\rho_b}_{\mathsf{A}} \otimes \ketbra{0}{0}_{\mathsf{BC}}$, where $\Pi = I_{\mathsf{A}} \otimes \ketbra{0}{0}_{\mathsf{BC}}$. 

    Next, by the linear-combination-of-unitaries technique \cite{CW12} and introducing a new subsystem $\mathsf{D}$ consisting of $1$ qubit, the following $\rbra{n+2a+1}$-qubit unitary operator
    \[
    U_\Delta = \rbra[\Big]{I_{\mathsf{ABC}} \otimes H_{\mathsf{D}}} \rbra[\Big]{ \rbra{U_0}_{\mathsf{ABC}} \otimes \ketbra{0}{0}_{\mathsf{D}} + \rbra{U_1}_{\mathsf{ABC}} \otimes \ketbra{1}{1}_{\mathsf{D}}} \rbra[\Big]{I_{\mathsf{ABC}} \otimes \rbra{HX}_{\mathsf{D}}}
    \]
    is a unitary dilation of $\rbra{\rho_0 - \rho_1}/2$ with $\rbra{2a+1}$ ancillary qubits, where $H$ is the Hadamard gate and $X$ is the Pauli-X gate. 
\end{proof}

\subsection{Quantum sample-to-query lifting}

We mention the quantum sample-to-query lifting method, which enables us to convert quantum query complexity to quantum sample complexity. 
This lifting method, used implicitly before in \cite{GP22,WZ24}, was pointed out in \cite{WZ23} and later refined in \cite{WZ24b,TWZ25}. 
Here, we use the version in \cite{TWZ25}. 

\begin{lemma}[Quantum sample-to-query lifting, adapted from {\cite[Theorem 1.5]{TWZ25}} and {\cite[Corollary 1.4]{CWZ25}}] \label{lemma:lifting}
    For any quantum algorithm $\mathcal{A}$ that uses $q$ queries to the state-preparation circuits of some unknown quantum states, there is another quantum algorithm $\mathcal{A}'$ using $O\rbra{q^2}$ samples of those unknown quantum states such that the output distribution of $\mathcal{A}'$ is $0.99$-close to that of $\mathcal{A}$. 
\end{lemma}

\section{Efficient quantum algorithms for estimating trace and operator norm distances when one state is pure}
\label{sec:one-state-pure}

This section focuses on the setting in which one of the two input states is pure. We first prove a structural warm-start property specific to this case in \Cref{subsec:constant-overlap-lemma}, then combine it with qubitization (\Cref{subsec:qubitization}) and maximum phase estimation to derive a quantum query algorithm in \Cref{subsec:query-one-state-pure}. Finally, the sample upper bound follows immediately in \Cref{subsec:sample-one-state-pure}.

\subsection{The input states provide a warm start when one state is pure}
\label{subsec:constant-overlap-lemma}

We begin with the structural statement that makes the single-pure-state setting easier:
\begin{lemma}[Single-pure-state warm-start lemma]
    \label{lemma:one-pure-state-warm-start}
    Let $\rho$ be a quantum state, and let $\ketbra{\psi}{\psi}$ be a pure state on the same Hilbert space. Define the Hermitian operator 
    $\Delta \coloneqq (\ketbra{\psi}{\psi}-\rho)/2.$ 
    Then either $\Delta = 0$, or $\Delta$ has the unique positive eigenvalue equal to $\norm{\Delta}_{\infty}$. If $\ket{v_+}$ is the unique unit eigenvector corresponding to this eigenvalue, then 
    \[ \abs*{\braket{v_+}{\psi}}^2 \geq \frac{1}{2}+\Abs{\Delta}_{\infty} > \frac{1}{2}.  \]
    Moreover, equality holds when $\rho$ is also pure and $\Delta\neq 0$. 
\end{lemma}

\begin{proof}
    For every pure state $\ket{\psi^\perp}$  orthogonal to $\ket{\psi}$, we have 
    \begin{equation}
        \label{eq:orthogonal-subspace-positive}
        \bra{\psi^\perp} \Delta \ket{\psi^\perp} = -\bra{\psi^\perp} \rho \ket{\psi^\perp}/2 \leq 0,
    \end{equation}
    since $\rho$ is positive semidefinite. Hence, the restriction of $\Delta$ to $\Pi_\perp$, namely 
    \[\Pi_\perp \Delta \Pi_\perp, \quad\text{where } \Pi_\perp \coloneqq I-\ketbra{\psi}{\psi},\] 
    is negative semidefinite. It follows that the positive eigenspace of $\Delta$ has dimension at most $1$: indeed, if $\ket{\phi_1}$ and $\ket{\phi_2}$ were two linearly independent positive-eigenvalue eigenvectors, then some nonzero linear combination of them would be orthogonal to $\ket{\psi}$, contradicting \Cref{eq:orthogonal-subspace-positive}.

    If $\Delta$ had no positive eigenvalue, then $\Delta$ is negative semidefinite, namely $\Delta\preceq 0$. Since $\tr(\Delta)=0$, this would force $\Delta=0$. Therefore, if $\Delta\neq 0$, then $\Delta$ has the unique positive eigenvalue, say $\tau > 0$. Since $\Delta$ is traceless, its negative eigenvalues sum to $-\tau$, and hence each of them is at least $-\tau$. Because $\Delta$ is Hermitian, we obtain: 
    \begin{equation}
        \label{eq:warm-start-neg-eigenval}
        \tau = \norm{\Delta}_\infty.
    \end{equation}
    
    Now let $\ket{v_+}$ be the unique unit eigenvector corresponding to the eigenvalue $\tau$. We have:
    \begin{equation}
        \label{eq:warm-start-decomp}
        \Delta\ket{v_+} = \tau\ket{v_+} \quad\text{and}\quad 2\cdot\Delta\ket{v_+} = \innerprod{\psi}{v_+} \ket{\psi}-\rho\ket{v_+}.
    \end{equation}
    
    Since $\rho$ is positive semidefinite and $\tr(\rho)=1$, we have $0 \preceq \rho \preceq I$, and hence $\rho^2 \preceq \rho$. Therefore, \Cref{eq:warm-start-decomp} implies the following:
    \begin{subequations}
    \label{eq:warm-start-squared-norm-a}
    \begin{align}
        \norm*{\rho \ket{v_+}}^2_2 &= \bra{v_+} \rho^2 \ket{v_+}\\
        &\leq \bra{v_+} \rho \ket{v_+} = \abs*{\innerprod{\psi}{v_+}}^2 - 2\tau.
    \end{align}
    \end{subequations}
    
    On the other hand, a direct calculation shows that
    \begin{subequations}
    \label{eq:warm-start-squared-norm-b}
    \begin{align}
        \norm*{\rho \ket{v_+}}^2_2 &= \norm*{\innerprod{\psi}{v_+} \ket{\psi}-2\tau \ket{v_+}}^2_2\\
        &= \abs*{\innerprod{\psi}{v_+}}^2 + 4\tau^2 - 4 \tau \abs*{\innerprod{\psi}{v_+}}^2.
    \end{align}
    \end{subequations}
    Consequently, combining \Cref{eq:warm-start-neg-eigenval,eq:warm-start-squared-norm-a,eq:warm-start-squared-norm-b} gives 
    \begin{equation}
        \label{eq:warm-start-bound}
        2\norm{\Delta}_\infty+1 = 2\tau + 1 \leq 2 \abs*{\innerprod{\psi}{v_+}}^2,
    \end{equation}
    which establishes the desired bound after rearranging the terms.
    
    To see the equality condition, note that if $\rho=\ketbra{\phi}{\phi}$ is pure, then the identity $\rho^2=\ketbra{\phi}{\phi}=\rho$, and it follows that both \Cref{eq:warm-start-squared-norm-a,eq:warm-start-bound} become identities, as desired.
\end{proof}

\subsection{Qubitization for unitary dilation}
\label{subsec:qubitization}

We introduce the following version of qubitization, as stated in \Cref{lemma:qubitization}, that converts eigenvalues of the Hermitian operator into eigenphases of an associated unitary dilation in a form suitable for phase estimation. Compared to the qubitization in \cite{LC19}, our version of qubitization requires \emph{no} ancillary qubits and produces a different rotation angle. 

\begin{lemma}[Qubitization for unitary dilation] \label{lemma:qubitization}
    Let $U$ be an $\rbra{n+a}$-qubit unitary dilation of an $n$-qubit Hermitian operator $H$. 
    Denote $\Pi = \ketbra{0}{0}^{\otimes a} \otimes I_n$, $R = 2\Pi - I_{n+a}$, and $W = URU^\dag R$, where $I_n$ denotes the $n$-qubit identity operator. 
    Let $\ket{v_j}$ be a unit eigenvector of $H$ with eigenvalue $\lambda_j \in \sbra{-1, 1}$, i.e., $H\ket{v_j} = \lambda_j\ket{v_j}$ for each $1 \leq j \leq 2^n$. 
    Denote $\ket{\psi_j} = \ket{0}^{\otimes a} \otimes \ket{v_j}$. 
    If $\abs{\lambda_j} < 1$, then define 
    \[
    \ket{\psi_j^\perp} \coloneqq \frac{U\ket{\psi_j} - \lambda\ket{\psi_j}}{\sqrt{1 - \lambda_j^2}}, \quad \ket{\phi_j^{\pm}} \coloneqq \frac{1}{\sqrt{2}} \rbra*{ \ket{\psi_j} \mp i \ket{\psi_j^\perp} },
    \]
    and the following holds:
    \begin{enumerate}
        \item $\Pi \ket{\psi_j} = \ket{\psi_j}$, $\Pi \ket{\psi_j^\perp} = 0$,  $U\ket{\psi_j} = \lambda_j\ket{\psi_j} + \sqrt{1 - \lambda_j^2} \ket{\psi_j^\perp}$,
        \item Let $\theta_j = \arccos\rbra{\lambda_j} \in \rbra{0, \pi}$, then 
        \begin{align*}
        W \ket{\psi_j} & = \cos\rbra{2\theta_j} \ket{\psi_j} + \sin\rbra{2\theta_j} \ket{\psi_j^\perp}, \\
        W \ket{\psi_j^\perp} & = -\sin\rbra{2\theta_j} \ket{\psi_j} + \cos\rbra{2\theta_j} \ket{\psi_j^\perp}, \\
        W \ket{\phi_j^{\pm}} & = e^{\pm i 2 \theta} \ket{\phi_j^{\pm}}.
        \end{align*}
    \end{enumerate}
    If $\abs{\lambda_j} = 1$, then $U \ket{\psi_j} = \lambda\ket{\psi_j}$ and $W\ket{\psi_j} = \ket{\psi_j}$. 

    Moreover, let $\mathcal{P} \coloneqq \operatorname{ran}\rbra{\Pi} = \spanspace\set{\ket{\psi_j}}{1 \leq j \leq 2^n}$, $\mathcal{Q} \coloneqq U \mathcal{P} = \operatorname{ran}\rbra{U \Pi U^\dag}$, and $\mathcal{M} = \rbra{\mathcal{P} + \mathcal{Q}}^{\perp}$, where $\operatorname{ran}\rbra{A} \coloneqq \set{A \ket{x}}{\ket{x} \in \operatorname{dom}\rbra{A}}$, then every nonzero vector $\ket{\varphi} \in \mathcal{M}$ is an eigenvector of $W$ with eigenvalue $1$, i.e., $W \ket{\varphi} = \ket{\varphi}$. 
\end{lemma}

\begin{proof}
    We first consider the case where $\abs{\lambda_j} < 1$. 
    \begin{enumerate}
        \item $\Pi\ket{\psi_j} = \rbra{\ketbra{0}{0}^{\otimes a} \otimes I_n} \rbra{\ket{0}^{\otimes a} \otimes \ket{v_j}} = \ket{0}^{\otimes a} \otimes \ket{v_j}  = \ket{\psi_j}$. 
        Hence, $\Pi U \ket{\psi_j} = \Pi U \Pi \ket{\psi_j} = H \ket{0}^{\otimes a}\otimes \ket{v_j} =  \ket{0}^{\otimes a}\otimes \lambda_j \ket{v_j} = \lambda_j \ket{\psi_j}$, which gives
        \begin{align*}
        \Pi \ket{\psi_j^\perp} 
        & = \frac{1}{\sqrt{1-\lambda_j^2}} \Pi\rbra*{U\ket{\psi_j} - \lambda_j\ket{\psi_j}} \\
        & = \frac{1}{\sqrt{1-\lambda_j^2}} \rbra*{ \Pi U\ket{\psi_j} - \lambda_j \Pi\ket{\psi_j} } \\
        & = \frac{1}{\sqrt{1-\lambda_j^2}} \rbra*{ \lambda_j \ket{\psi_j} - \lambda_j \ket{\psi_j} } \\
        & = 0.
        \end{align*}
        Moreover, $\lambda_j\ket{\psi_j} + \sqrt{1 - \lambda_j^2} \ket{\psi_j^\perp} = \lambda_j \ket{\psi_j} + \rbra*{U \ket{\psi_j} - \lambda_j\ket{\psi_j}} = U \ket{\psi_j}$.

        \item Since $U^\dag$ is also a unitary dilation of $H$ (note that $H^\dag = H$), we have
        \[
        U^\dag \ket{\psi_j} = \lambda_j\ket{\psi_j} + \sqrt{1-\lambda_j^2} \ket{\psi^\perp_{j,\dag}}
        \]
        for some $\ket{\psi^\perp_{j,\dag}}$ satisfying $\Pi\ket{\psi^\perp_{j,\dag}} = 0$.
        Then, by direct calculation, we have
        \begin{align*}
            W \ket{\psi_j}
            & = URU^\dag R \ket{\psi_j} \\
            & = URU^\dag \ket{\psi_j} \\
            & = UR \rbra*{\lambda_j\ket{\psi_j} + \sqrt{1-\lambda_j^2} \ket{\psi^\perp_{j,\dag}}} \\
            & = U \rbra*{\lambda_j\ket{\psi_j} - \sqrt{1-\lambda_j^2} \ket{\psi^\perp_{j,\dag}}} \\
            & = \lambda_j U \ket{\psi_j} - \sqrt{1-\lambda_j^2} U \ket{\psi^\perp_{j,\dag}} \\
            & = \lambda_j U \ket{\psi_j} - U \rbra*{ U^\dag\ket{\psi_j} - \lambda_j\ket{\psi_j} } \\
            & = 2 \lambda_j U \ket{\psi_j} - \ket{\psi_j} \\
            & = 2\lambda_j\rbra*{\lambda_j\ket{\psi_j} + \sqrt{1 - \lambda_j^2} \ket{\psi_j^\perp}} - \ket{\psi_j} \\
            & = \rbra*{2\lambda_j^2 - 1} \ket{\psi_j} + 2 \lambda_j \sqrt{1 - \lambda_j^2} \ket{\psi_j^\perp} \\
            & = \rbra*{2\cos^2\rbra{\theta_j} - 1} \ket{\psi_j} + 2 \cos\rbra{\theta_j} \sqrt{1 - \cos^2\rbra{\theta_j}} \ket{\psi_j^\perp} \\
            & = \cos\rbra{2\theta_j} \ket{\psi_j} + \sin\rbra{2\theta_j} \ket{\psi_j^\perp}.
        \end{align*}
        Similarly, 
        \begin{align*}
            W \ket{\psi_j^\perp}
            & = URU^\dag R \ket{\psi_j^\perp} \\
            & = - URU^\dag \ket{\psi_j^\perp} \\
            & = - URU^\dag \cdot \frac{U\ket{\psi_j} - \lambda_j\ket{\psi_j}}{\sqrt{1 - \lambda_j^2}} \\
            & = - \frac{1}{\sqrt{1-\lambda_j^2}} U\ket{\psi_j} + \frac{\lambda_j}{\sqrt{1-\lambda_j^2}} URU^\dag \ket{\psi_j} \\
            & = - \frac{1}{\sqrt{1-\lambda_j^2}} \rbra*{\lambda_j\ket{\psi_j} + \sqrt{1 - \lambda_j^2} \ket{\psi_j^\perp}} + \frac{\lambda_j}{\sqrt{1-\lambda_j^2}} \rbra*{\rbra*{2\lambda_j^2 - 1} \ket{\psi_j} + 2 \lambda_j \sqrt{1 - \lambda_j^2} \ket{\psi_j^\perp}} \\
            & = - 2\lambda_j \sqrt{1 - \lambda_j^2} \ket{\psi} + \rbra*{2\lambda_j^2 - 1} \ket{\psi_j^\perp} \\
            & = - 2 \cos\rbra{\theta_j} \sqrt{1 - \cos^2\rbra{\theta_j}} \ket{\psi_j} + \rbra*{2\cos^2\rbra{\theta_j} - 1} \ket{\psi_j^\perp} \\
            & = - \sin\rbra{2\theta_j} \ket{\psi_j} + \cos\rbra{2\theta_j} \ket{\psi_j^\perp}. 
        \end{align*}
        With the above identities, it can also be verified that $W \ket{\phi_j^{\pm}} = e^{\pm i 2 \theta} \ket{\phi_j^{\pm}}$.
    \end{enumerate}

    Now we consider the case where $\abs{\lambda_j} = 1$.
    Previously, we have shown that $\Pi U \ket{\psi_j} = \lambda_j\ket{\psi_j}$. 
    Hence, $U\ket{\psi_j} = \lambda_j\ket{\psi_j} + \ket{\psi_j^\perp}$ for some $\ket{\psi_j^\perp}$ satisfying $\Pi \ket{\psi_j^\perp} = 0$. 
    On the other hand, $1 = \Abs{U\ket{\psi_j}}^2 = \Abs[\big]{\lambda_j\ket{\psi_j} + \ket{\psi_j^\perp}}^2 = \abs{\lambda_j}^2 + \Abs{\ket{\psi_j^\perp}}^2 = 1 + \Abs{\ket{\psi_j^\perp}}^2$, which gives $\Abs{\ket{\psi_j^\perp}} = 0$. 
    Therefore, $U\ket{\psi_j} = \lambda_j\ket{\psi_j}$, i.e., $\ket{\psi_j}$ is an eigenvector of $U$ with eigenvalue $\lambda_j$. 
    This implies that $\ket{\psi_j}$ is an eigenvector of $U^\dag$ with eigenvalue $\lambda_j^* = \lambda_j$ (note that $\lambda_j \in \mathbb{R}$), i.e., $U^\dag \ket{\psi_j} = \lambda_j\ket{\psi_j}$. 
    Therefore,
    \begin{align*}
        W \ket{\psi_j}
        & = U R U^\dag R \ket{\psi_j} \\
        & = U R U^\dag \ket{\psi_j} \\
        & = U R \rbra{\lambda_j \ket{\psi_j}} \\
        & = U \rbra*{\lambda_j \ket{\psi_j}} \\
        & = \lambda_j^2 \ket{\psi_j} \\
        & = \ket{\psi_j}. 
    \end{align*}

    For any nonzero vector $\ket{\varphi} \in \mathcal{M}$, it holds that $\ket{\varphi} \perp \mathcal{P}$ and $\ket{\varphi} \perp \mathcal{Q}$, which implies that $\Pi \ket{\varphi} = 0$ and $\Pi U^\dag \ket{\varphi} = 0$. 
    Therefore, 
    \begin{align*}
        W \ket{\varphi}
        & = U R U^\dag R \ket{\varphi} \\
        & = U R U^\dag \rbra{2\Pi - I} \ket{\varphi} \\
        & = U R U^\dag \rbra{-\ket{\varphi}} \\
        & = U \rbra{I - 2\Pi} U^\dag \ket{\varphi} \\
        & = \rbra{I - 2U\Pi U^\dag} \ket{\varphi} \\
        & = \ket{\varphi},
    \end{align*}
    which means that $\ket{\varphi}$ is an eigenvector of $W$ with eigenvalue $1$.
\end{proof}

\subsection{Quantum query algorithm when one state is pure}
\label{subsec:query-one-state-pure}

Beyond presenting the main result of this section, we need the following lemmas: 

\begin{lemma}[Quantum maximum phase estimation~\protect{\cite[Lemma 4.3]{MdW23}}]\label{lem:max_phase_estimation}
Let $\gamma \in \rbra{0, 1}$.
Suppose we have a $d$-dimensional unitary $U$ with unknown eigenvectors $\ket{\phi_1},\dots,\ket{\phi_d}$ and associated eigenphases $\theta_1,\dots, \theta_d \in [0, 2\pi - 2\epsilon)$. Suppose we also have a state preparation unitary $S$ such that
    \[
    S \ket{0} = \sum_{j=1}^d \alpha_j \ket{\phi_j}\ket{\varphi_j},
    \]
    where $\sum_{j:\theta_j=\theta_{\max}}\abs{\alpha_j}^2 \geq \gamma^2$ and $\ket{\varphi_j}$ are arbitrary (normalized) states. Then there exists a quantum algorithm $\mathsf{MaxPhaseEst}(U, S, \epsilon, \gamma)$ that uses at most $O(1/\gamma)$ queries to $S$, and $O(\log(1/\gamma)/\gamma\epsilon)$ queries to $U$, and with probability at least $0.99$ outputs $\tilde \theta$ such that $\abs{\tilde \theta - \theta_{\max}} \leq \epsilon$.
    Moreover, if $\sum_{j:\theta_j=\theta_{\max}}\abs{\alpha_j}^2 < \gamma^2$, then with probability at least $0.99$, $\abs{\tilde\theta - \theta_j} \leq \epsilon$ for some $j$.
\end{lemma}

\begin{lemma}\label{lem:T_inf_first_pure}
    For a pure state $\ketbra{\psi}{\psi}$ and a state $\rho$, it holds that
    \[
    \Tinfty\ab(\ketbra{\psi}{\psi}, \rho) = \lambda_{\max}\ab(\frac{\ketbra{\psi}{\psi} - \rho}{2}).
    \]
    More generally, we have
    \[
    \lambda_{\max}\ab(\frac{\ketbra{\psi}{\psi} - \rho}{2}) \geq \lambda_{\max}\ab(\frac{\rho- \ketbra{\psi}{\psi}}{2}).
    \]
\end{lemma}
\begin{proof}
Let $\Delta \coloneq (\ketbra{\psi}{\psi}-\rho)/2$. If $\Delta = 0$, the lemma holds trivially. If $\Delta \neq 0$, by \cref{lemma:one-pure-state-warm-start}, $\Delta$ has a unique positive eigenvalue equal to $\norm{\Delta}_{\infty}$. The sum of all negative eigenvalues of $\Delta$ (with multiplicities) is equal to $-\norm{\Delta}_{\infty}$ due to $\tr\rbra{\Delta} = 0$.
Therefore, no negative eigenvalue is less than $-\norm{\Delta}_{\infty}$, which implies that $\lambda_{\max}\ab(\Delta) \geq \lambda_{\max}\ab(-\Delta)$.
\end{proof}

\subsubsection{The quantum query algorithm}
Now we establish the main result of this section: 

\begin{theorem}\label{thm:query-pure}
    Given quantum query access to the state-preparation circuits $Q_0$ and $Q_1$ of quantum states $\rho_0$ and $\rho_1$ with the promise that one of the two states is pure, there exists an explicit quantum algorithm that estimates $\td\rbra{\rho_0, \rho_1}$ and $\Tinfty\rbra{\rho_0, \rho_1}$ to within additive error $\epsilon$, using $O\ab(1/\epsilon)$ queries to $Q_0$ and $Q_1$.
\end{theorem}

\begin{proof}
Given $(n+a)$-qubit unitary oracles $Q_0$ and $Q_1$ that prepare purifications $\ket{\psi_0}$ and $\ket{\psi_1}$ of the $n$-qubit quantum states $\rho_0$ and $\rho_1$, respectively. 
We can estimate $\td\rbra{\rho_0, \rho_1}$ and $\Tinfty\rbra{\rho_0, \rho_1}$ as follows.

\begin{enumerate}
    \item For each $k \in \cbra{0, 1}$, 
\begin{enumerate}[label=1.\arabic*.]
    \item Construct a unitary operator $U_{A_k}$ that is an $(n+2a+2)$-qubit unitary dilation of
    \[
    A_k \coloneq \frac{I - (-1)^{k}\Delta}{2}, \qquad \text{where } \Delta \coloneq \frac{\rho_0-\rho_1}{2},
    \]
    using $O(1)$ queries to $Q_0$ and $Q_1$ by \cref{lemma:dilation-Delta} and linear combination of unitaries~\cite{CW12}.
    \item Construct a qubitization unitary operator $W_{A_k} \coloneq U_{A_k}R(U_{A_k})^\dag R$, where $R = 2\Pi - I_{n+2a+2}$ and $\Pi = \ketbra{0}{0}^{\otimes (2a+2)} \otimes I_n$, using $O(1)$ queries to $U_{A_k}$. 
    \item Apply the maximum phase estimation algorithm in \cref{lem:max_phase_estimation}, let $\tilde{\theta}_{\max}^{(k)}$ be the output of $\mathsf{MaxPhaseEst}\ab*(-W_{A_k},\, I_{2a+2} \otimes Q_k,\, \epsilon,\, \gamma)$, where $\gamma = 1/2$.
\end{enumerate}
    \item Let $\tilde\theta_{\max} = \max\ab*\{\tilde{\theta}_{\max}^{(0)}, \tilde{\theta}_{\max}^{(1)}\}$. 
    \item Return $\tilde{\lambda}_{\max} = 1 - 2\cos\ab*((\pi - \tilde{\theta}_{\max})/{2})$.
\end{enumerate}
The above algorithm is formally given in \cref{alg:T_inf_estimation_one_pure}.
\begin{algorithm}[ht]
\caption{Operator norm distance estimation when one state is pure.}
\label{alg:T_inf_estimation_one_pure}
\begin{algorithmic}[1]
\Require $(n+a)$-qubit unitary oracles $Q_0$ and $Q_1$ that prepare purifications of $n$-qubit quantum states $\rho_0$ and $\rho_1$, precision $\epsilon \in (0, 0.1)$.
\Ensure An estimate of $\Tinfty(\rho_0, \rho_1)$.
 
\State Denote $\Delta \coloneq (\rho_0-\rho_1)/2$ and $R \coloneq 2\Pi - I_{n+2a+2}$, where $\Pi \coloneq \ketbra{0}{0}^{\otimes(2a+2)} \otimes I_n$.

\For{$k \in \{0,1\}$}
    \State Construct an $(n+2a+2)$-qubit unitary dilation $U_{A_k}$ of
    $A_k \coloneq \rbra{I-(-1)^k\Delta}/{2}$.
    \State Construct a qubitization unitary $W_{A_k} \coloneq U_{A_k} R \rbra{U_{A_k}}^\dagger R$.
    \State $\tilde{\theta}^{(k)}_{\max} \gets \mathsf{MaxPhaseEst}\ab*(-W_{A_k},\, I_{2a+2} \otimes Q_k,\, \epsilon/2,\, \gamma)$, where $\gamma = 1/2$.
\EndFor
\State $\tilde\theta_{\max} \gets \max\ab*\{\tilde{\theta}_{\max}^{(0)}, \tilde{\theta}_{\max}^{(1)}\}$.
\State \Return $\tilde{\lambda}_{\max} \gets 1 - 2\cos\ab*((\pi - \tilde{\theta}_{\max})/{2})$.
\end{algorithmic}
\end{algorithm}

First, we analyze the complexity of \cref{alg:T_inf_estimation_one_pure}. 
Note that for each $k \in \cbra{0, 1}$, the unitary dilation $U_{A_k}$ can be implemented using $O(1)$ queries to $Q_0$ and $Q_1$ by \cref{lemma:dilation-Delta} and linear combination of unitaries~\cite{CW12}. 
The qubitization unitary $-W_{A_k}$ uses $O(1)$ queries to $U_{A_k}$. 
The subroutine $\mathsf{MaxPhaseEst}(-W_{A_k}, I_{2a+2}\otimes Q_k, \epsilon, \gamma)$ uses a total number of $O(1/\epsilon)$ queries to $Q_k$ and $U_{A_k}$. Therefore, the total number of queries to $Q_0$ and $Q_1$ is $O(1/\epsilon)$. 

Next, we prove the correctness of the algorithm. Without loss of generality, we may assume that $\rho_0 = \ketbra{\psi}{\psi}$ is pure (if $\rho_1$ is pure, the analysis is similar). 
We now consider the properties of $\tilde\theta_{\max}^{\rbra{k}}$ for each $k \in \cbra{0, 1}$.
For $k = 0$, note that the random variable $\tilde{\theta}^{(0)}_{\max}$ is the output of $ \mathsf{MaxPhaseEst}\ab*(-W_{A_0},\, I_{2a+2} \otimes Q_0,\, \epsilon,\, \gamma)$, where $\gamma = 1/2$. Then, it holds (shown later) that
\begin{equation}\label{eq:subroutine_ouput_0}
    \Pr\sbra[\Big]{\abs[\big]{\tilde{\theta}^{(0)}_{\max} - \theta_{\max} }\leq \epsilon} \geq 0.99, \qquad \text{where } \theta_{\max} = \pi + 2\arccos\ab(\frac{1-\lambda_{\max}(\Delta)}{2}).
\end{equation}
For $k = 1$, the random variable $\tilde{\theta}^{(1)}_{\max}$ is the output of $ \mathsf{MaxPhaseEst}\ab*(-W_{A_1},\, I_{2a+2} \otimes Q_1,\, \epsilon,\, \gamma)$, where $\gamma = 1/2$. Then, with probability at least $0.99$, there exists an eigenphase $\theta'$ of $-W_{A_1}$ such that $\abs{\tilde{\theta}^{(1)}_{\max} - \theta' }\leq \epsilon$.
Moreover, we have (shown later)
\begin{equation} \label{eq:theta}
    \theta' \leq \theta_{\max}.
\end{equation}

Therefore, when we take $\tilde{\theta}_{\max} = \max\ab*\{\tilde{\theta}_{\max}^{(0)}, \tilde{\theta}_{\max}^{(1)}\}$, it holds that
\[
\Pr\sbra*{\abs[\big]{\tilde{\theta}_{\max} - \theta_{\max} }\leq \epsilon} \geq 0.99^2 > 0.98.
\]
With this, we can bound the error of the estimate $\tilde\lambda_{\max}$ as:
\begin{equation} \label{eq:error}
\Pr\sbra*{\abs*{\tilde\lambda_{\max} - \Tinfty\rbra{\rho_0, \rho_1}} \leq \epsilon} > 0.98.
\end{equation}
To see this, note that $\Tinfty(\rho_0, \rho_1) = \lambda_{\max}(\Delta)$ when $\rho_0$ is pure by \cref{lem:T_inf_first_pure}.
Therefore, by the mean value theorem,
\begin{align*}
    \abs*{\tilde{\lambda}_{\max} - \Tinfty(\rho_0, \rho_1)} = \abs*{\tilde{\lambda}_{\max} - \lambda_{\max}(\Delta)} = \abs*{\sin\ab(\frac{\xi-\pi}{2})} \abs*{\tilde{\theta}_{\max} - \theta_{\max}}.
\end{align*}
for some $\xi$ between $\theta_{\max}$ and $\tilde{\theta}_{\max}$, which gives \cref{eq:error} by noting that $\abs{\sin\ab((\xi-\pi)/2)} \leq 1$.
Consequently, \cref{alg:T_inf_estimation_one_pure} returns an estimate of $\Tinfty\rbra{\rho_0,\rho_1}$ to within additive error $\epsilon$ with probability at least $0.98$.
We have proved that \cref{alg:T_inf_estimation_one_pure} estimates $\Tinfty(\rho_0, \rho_1)$ to within additive error $\epsilon$. 
In addition, since $\td(\rho_0, \rho_1) = 2\Tinfty(\rho_0, \rho_1)$, the algorithm can estimate $\td(\rho_0, \rho_1)$ to within additive error $\epsilon$ by estimating $\Tinfty(\rho_0, \rho_1)$ to within additive error $\epsilon/2$.

To complete the proof, it remains to prove \cref{eq:subroutine_ouput_0,eq:theta}.
We begin with the analysis on the subroutine $\mathsf{MaxPhaseEst}(-W_{A_0}, I_{2a+2} \otimes Q_0, \epsilon, \gamma)$. To show that applying the quantum maximum phase estimation algorithm in \cref{lem:max_phase_estimation} gives a result satisfying \cref{eq:subroutine_ouput_0}, we prove the following properties in order.
\begin{enumerate}
    \item The eigenphases of $-W_{A_0}$ are in $[0, 2\pi - 2\epsilon)$, and the largest eigenphase of $-W_{A_0}$ is $\theta_{\max} = \pi + 2\arccos\ab((1-\lambda_{\max}(\Delta))/{2})$.
    \item The state preparation unitary $I_{2a+2} \otimes Q_0$ prepares a state that has an overlap at least $\gamma = 1/2$ with the top eigenspace of $-W_{A_0}$.
\end{enumerate}

Let $d=2^n$ be the dimension of the quantum states $\rho_0$ and $\rho_1$. For $j\in \{1,\dots, d\}$, let $\ket{v_j}$ denote an eigenvector of $\Delta$ with eigenvalue $\lambda_j$, i.e., $\Delta\ket{v_j} = \lambda_j \ket{v_j}$, then $\ket{v_j}$ is also an eigenvector of $A_0$ such that
\[
A_0 \ket{v_j} = \mu_j \ket{v_j}, \qquad \mu_j \coloneq \frac{1-\lambda_j}{2}.
\]
Note that the eigenvalues of $\Delta$ are in $[-1/2, 1/2]$, and therefore the eigenvalues of $A_0$ are in $[1/4, 3/4]$.

Let $\ket{u_j} = \ket{0}^{\otimes \rbra{2a+2}}\otimes\ket{v_j}$, and define 
\[
    \ket{u_j^\perp} \coloneq \frac{U_{A_0}\ket{u_j} - \mu_j\ket{u_j}}{\sqrt{1 - \mu_j^2}}.
\]
By \cref{lemma:qubitization}, since $\abs{\mu_j}<1$, the qubitization unitary $W_{A_0}$ has eigenvectors 
\begin{align*}
        \ket{\phi^{\pm}_{j}} & = \frac{1}{\sqrt{2}} \rbra*{ \ket{u_j} \mp i \ket{u_j^\perp}}, 
\end{align*}
with $W_{A_0} \ket{\phi^{\pm}_{j}} = e^{\pm i \vartheta_j} \ket{\phi^{\pm}_{j}}$, where
\[
\vartheta_j \coloneq 2\arccos(\mu_j) = 2\arccos\ab(\frac{1-\lambda_j}{2}).
\]
Moreover, let $\mathcal{P} \coloneqq \spanspace\set{\ket{u_j}}{1 \leq j \leq 2^n}$, $\mathcal{Q} \coloneq \spanspace\set{\ket{u_j^\perp}}{1 \leq j \leq 2^n}$, and $\mathcal{M} = \rbra{\mathcal{P} + \mathcal{Q}}^{\perp}$, then every nonzero vector $\ket{\phi'} \in \mathcal{M}$ is an eigenvector of $W_{A_0}$ with eigenphase $0$, i.e., $W_{A_0} \ket{\phi'} = \ket{\phi'}$. 

Since the function $\arccos((1-x)/{2})$ is strictly monotonically increasing on $[-1/2, 1/2]$, for increasing eigenvalues of $\Delta$, 
\[
\lambda_1 \leq \lambda_2 \leq \cdots \leq \lambda_{d},
\]
we have
\[
\vartheta_1 \leq \vartheta_2 \leq \cdots \leq \vartheta_{d}, \quad\text{where }
\vartheta_{j} \in \ab[2\arccos\ab(3/4), 2\arccos\ab(1/4)] \subset (0, \pi-2\epsilon).
\]
Now we add a global phase $e^{i\pi}$ to $W_{A_0}$, i.e., $e^{i\pi}W_{A_0} = - W_{A_0}$. The phases of $- W_{A_0}$, denoted by $\theta^{\pm}_{j} = \pi \pm \vartheta_j$, lie in $[0, 2\pi-2\epsilon)$ as $\pi - \vartheta_j \in (0,\pi)$ and $\pi + \vartheta_j \in (\pi,2\pi-2\epsilon)$. The largest phase of $- W_{A_0}$, denoted by $\theta_{\max} \coloneq \max_j\{\theta^{\pm}_{j}\} = \max_j\{\theta^{+}_{j}\}$, clearly corresponds to the largest eigenvalue of $\Delta$, 
\begin{equation}\label{eq:max_eigenphase_onepure}
    \theta_{\max}= \pi+\vartheta_{N} = \pi + 2\arccos\ab(\frac{1-\lambda_{\max}(\Delta)}{2}).
\end{equation}

Then we show that the state preparation unitary $I_{2a+2} \otimes Q_0$ prepares a state that has an overlap at least $\gamma = 1/2$ with the top eigenspace of $-W_{A_0}$. The unitary $I_{2a+2} \otimes Q_0$ prepares the input state
\[
\ket{\Psi_0} \coloneq (I_{2a+2} \otimes Q_0) \ket{0}^{\otimes \rbra{n+3a+2}} = \ket{0}^{\otimes \rbra{2a+2}} \ket{\psi_0} = \ket{0}^{\otimes \rbra{2a+2}} \ket{\psi}_{\sfA}\ket{\varphi}_{\sfB},
\]
where $\ket{\psi_0} = \ket{\psi}_{\sfA}\ket{\varphi}_{\sfB}$ is the purification of the pure state $\rho = \ketbra{\psi}{\psi}$.
\cref{lemma:one-pure-state-warm-start} shows that either $\Delta = 0$, or $\Delta$ has the unique positive eigenvalue equal to $\norm{\Delta}_{\infty}$. Note that if $\Delta=0$, then $\lambda_j = 0$ and $\theta_j = 5\pi/3$ for all $j$. In this case, $\ket{\psi}$ is an eigenstate of $A = I/2$ with the maximum eigenvalue $1/2$, thus it trivially holds that the input state $\ket{\Psi_0}$ has an overlap at least $\gamma=1/2$ with the top eigenspace of $-W_{A_0}$.

If $\Delta \neq 0$, let $\ket{v_+}$ be the unique unit eigenvector corresponding to this eigenvalue, then 
\[
\abs*{\braket{v_+}{\psi}}^2 \geq \frac{1}{2}. \] 
The overlap between the input state and the top eigenspace of $-W_{A_0}$ is at least
\begin{align*}
    \sum_{j:\theta^{+}_{j} = \theta_{\max}} \norm[\big]{(\bra{\phi_{j}^{+}}\otimes I_a) \ket{\Psi_0}}_2^2 &= \sum_{j:\theta^{+}_{j} = \theta_{\max}} \norm[\Big]{\ab\Big(\frac{1}{\sqrt{2}} \ab\big( \bra{u_j} + i \bra{u_j^\perp})\otimes I_a) \ket{\Psi_0}}_2^2 \\
    &\geq \sum_{j:\theta^{+}_{j} = \theta_{\max}} \norm[\Big]{\frac{1}{\sqrt{2}}\ab\big(\bra{u_j} \otimes I_a) \ket{\Psi_0}}_2^2\\
    &= \sum_{j:\lambda_{j}=\lambda_{\max}(\Delta)} \norm[\Big]{\frac{1}{\sqrt{2}}\big( \bra{0}^{\otimes {2a+2}} \otimes \bra{v_j}  \otimes I_a)\ab\big(\ket{0}^{\otimes \rbra{2a+2}} \ket{\psi}_{\sfA}\ket{\varphi}_{\sfB})}_2^2\\
    &= \sum_{j:\lambda_{j}=\lambda_{\max}(\Delta)} \norm[\Big]{\frac{1}{\sqrt{2}}\big( \bra{v_j}  \otimes I_a)\ab\big(\ket{\psi}_{\sfA}\ket{\varphi}_{\sfB})}_2^2\\
    &\geq \frac{1}{2}\abs*{\braket{v_+}{\psi}}^2\\
    &\ge \frac{1}{4}.
\end{align*}
Thus we can take $\gamma = \sqrt{1/4} =1/2$. Therefore, by \cref{lem:max_phase_estimation}, the maximum phase estimation algorithm $\mathsf{MaxPhaseEst}(-W_{A_0}, I_{2a+2} \otimes Q_0, \epsilon, \gamma)$ uses $O(1)$ queries to $Q_0$ and $O(1/\epsilon)$ queries to $-W_{A_0}$, outputs an estimate $\tilde{\theta}_{\max}^{(0)}$ satisfying \cref{eq:subroutine_ouput_0}.

Now we prove \cref{eq:theta}.
Considering the second subroutine $k=1$, since $\rho_1$ could be an arbitrary mixed state, the input state prepared by the unitary $I_{2a+2} \otimes Q_1$ is not guaranteed to have overlap $\gamma = \frac{1}{2}$ with the top eigenspace of $-W_{A_1}$. However, even when the overlap is not promised, with probability at least $0.99$, the output of the algorithm $\mathsf{MaxPhaseEst}(-W_{A_1}, I_{2a+2} \otimes Q_1, \epsilon, 1/2)$ is within $\epsilon$ of an eigenphase of $-W_{A_1}$, whose corresponding eigenvector appears in the spectral decomposition of $I_{2a+2} \otimes Q_1 \ket{0}^{\otimes(n+3a+2)}$. By \cref{lem:T_inf_first_pure}, we have $\lambda_{\max}(\Delta) \geq \lambda_{\max}(-\Delta)$, which implies that all phases of $-W_{A_1}$ are less than or equal to the largest phase of $-W_{A_0}$, that is, for any eigenphase $\theta'$ of $-W_{A_1}$, we have $\theta' \leq \theta_{\max}$.
\end{proof}

\subsection{Quantum sample algorithm  when one state is pure}
\label{subsec:sample-one-state-pure}

\begin{theorem} \label{thm:sample-pure}
    When one of the quantum states $\rho_0$ and $\rho_1$ is pure, there is a quantum algorithm that estimates $\td\rbra{\rho_0, \rho_1}$ and $\Tinfty\rbra{\rho_0, \rho_1}$ to within additive error $\epsilon$, using $O\rbra{1/\epsilon^2}$ samples of $\rho_0$ and $\rho_1$. 
\end{theorem}

\begin{proof}
    This is obtained immediately by \cref{lemma:lifting} using the query complexity $O\rbra{1/\epsilon}$ given in \cref{thm:query-pure}, which gives the sample complexity $O\rbra{1/\epsilon^2}$. 
\end{proof}

\section{Efficient quantum algorithms for estimating operator norm distance}

We now turn to the general setting, in which neither input state is assumed to be pure. The same overall strategy still applies, but the constant-overlap warm start from \Cref{sec:one-state-pure} is no longer available and therefore must be replaced by a weaker overlap guarantee. In particular, by combining the overlap lemma in \Cref{subsec:general-overlap-lemma} with the qubitization framework from \Cref{subsec:qubitization}, we obtain general quantum query and sample algorithms for estimating the operator norm distance, which are presented respectively in \Cref{subsec:query-general,subsec:general-sample}.

\subsection{The input states provide an eigenvalue-scaled overlap}
\label{subsec:general-overlap-lemma}

To handle arbitrary mixed states, we next prove a substitute for the warm-start lemma by working with purifications of the input states. The resulting overlap bound is no longer constant, but it scales with the top eigenvalue and is sufficient for the general algorithm.

\begin{lemma}[Top eigenspace overlap lemma]\label{lem:top_eigenspace_overlap}
Let $\ket{\psi_0}$ and $\ket{\psi_1}$ be $(n+a)$-qubit purifications of $n$-qubit quantum states $\rho_0$ and $\rho_1$, respectively. Define $\Delta \coloneq (\rho_0 - \rho_1)/2$ and let $\Pi_{\max}^{(\Delta)}$ be the projector onto the top eigenspace of $\Delta$, which corresponds to its maximum eigenvalue $\lambda_{\max}(\Delta)$. Then, it holds that:
\[ \norm[\big]{\ab\big(  \Pi_{\max}^{(\Delta)} \otimes I_{a})\ket{\psi_0}}_2^2 \ge 2\lambda_{\max}(\Delta).\]
Moreover, if $\lambda_{\max}(\Delta)>0$, then equality holds if and only if 
\[ \rank\rbra[\big]{\Pi_{\max}^{(\Delta)}} = 1 \quad\text{and}\quad \Pi_{\max}^{(\Delta)} \rho_1 \Pi_{\max}^{(\Delta)} = 0.\]
In particular, if $\rho_0 = \ketbra{\phi}{\phi}$ is pure, then equality holds if and only if $\rho_1$ and $\ket{\phi}$ are orthogonal.
\end{lemma}

\begin{proof}
Since $\ket{\psi_0}$ and $\ket{\psi_1}$ are purifications of $\rho_0$ and $\rho_1$, respectively, for any $n$-qubit operator $M$ acting on the system register, it follows that:
\[
\bra{\psi_i}( M \otimes I_a )\ket{\psi_i}
=
\tr(M\rho_i),
\qquad i\in\{0,1\}.
\]
Therefore, we have
\begin{subequations}
\label{eq:top-overlap-avg}
\begin{align}
\norm[\big]{\ab\big( \Pi_{\max}^{(\Delta)}\otimes I_a)\ket{\psi_0}}_2^2 &= \bra{\psi_0}\ab\big(\Pi_{\max}^{(\Delta)}\otimes I_a )\ket{\psi_0}\\
&= \tr\ab\big(\Pi_{\max}^{(\Delta)}\rho_0).
\end{align}
\end{subequations}

On the other hand, since $\Pi_{\max}^{(\Delta)}$ projects onto the top eigenspace of $\Delta$, which corresponds to the eigenvalue $\lambda_{\max}$, we have
\[
\Pi_{\max}^{(\Delta)} \Delta = \Delta \Pi_{\max}^{(\Delta)} = \lambda_{\max}(\Delta)\Pi_{\max}^{(\Delta)}.
\]
Taking the trace gives
\begin{equation}
\label{eq:trace-top-eigspace}
\tr\ab\big(\Pi_{\max}^{(\Delta)}\Delta) = \lambda_{\max}(\Delta)\tr\ab\big(\Pi_{\max}^{(\Delta)}).
\end{equation}

Since $\rho_1\succeq 0$ and $\Pi_{\max}^{(\Delta)}$ is a projector, we have $\tr\ab\big(\Pi_{\max}^{(\Delta)}\rho_1)\ge 0$, and hence
\begin{subequations}
\label{eq:avg-ge-lambda}
\begin{align}
\tr\ab\big(\Pi_{\max}^{(\Delta)}\rho_0)
&\ge \tr\ab\big(\Pi_{\max}^{(\Delta)}\rho_0) -  \tr\ab\big(\Pi_{\max}^{(\Delta)}\rho_1)\\
&= \tr\ab\big(\Pi_{\max}^{(\Delta)}(\rho_0-\rho_1))\\
&= 2\tr\ab\big(\Pi_{\max}^{(\Delta)}\Delta)\\
&= 2\lambda_{\max}(\Delta)\tr\ab\big(\Pi_{\max}^{(\Delta)})\\
&\ge 2\lambda_{\max}(\Delta).
\end{align}
\end{subequations}
Here, the third line follows by substituting $\Delta=(\rho_0-\rho_1)/2$ into \Cref{eq:trace-top-eigspace}, and the last line uses the fact that $\tr\ab\big(\Pi_{\max}^{(\Delta)}) = \rank\ab\big(\Pi_{\max}^{(\Delta)})\ge 1$, because $\Pi_{\max}^{(\Delta)}\neq 0$.

Combining \cref{eq:top-overlap-avg} and \cref{eq:avg-ge-lambda} yields the desired overlap bound: 
\[
\norm[\big]{\ab\big(\Pi_{\max}^{(\Delta)}\otimes I_a)\ket{\psi_0}}_2^2 \ge 2\lambda_{\max}(\Delta). 
\]

To see the equality condition in the case $\lambda_{\max}(\Delta)>0$, we examine when both inequalities in \Cref{eq:avg-ge-lambda} are tight, which corresponds to the following conditions 
\begin{equation}
    \label{eq:overlap-lemma-equal-cond}
    \tr\ab\big(\Pi_{\max}^{(\Delta)}\rho_1)=0 \quad\text{and}\quad \tr\ab\big(\Pi_{\max}^{(\Delta)})=1.
\end{equation} 
Therefore, the first condition means that $\rho_1$ has no support on the top eigenspace of $\Delta$, equivalently $\Pi_{\max}^{(\Delta)} \rho_1 \Pi_{\max}^{(\Delta)} = 0$; while the second condition means that $\rank\rbra[\big]{\Pi_{\max}^{(\Delta)}} = 1$.

In particular, when $\rho_0$ is pure, \Cref{eq:overlap-lemma-equal-cond} simplifies to the condition that $\rho_0=\ketbra{\psi}{\psi}$ is orthogonal to $\rho_1$, more precisely, $\bra{\psi}\rho_1\ket{\psi}=0$.
\end{proof}

\subsection{Quantum query algorithm for estimating operator norm  distance}
\label{subsec:query-general}

We now present the main result of this section: 
\begin{theorem} \label{thm:query-general}
    Given quantum query access to the state-preparation circuits $Q_0$ and $Q_1$ of quantum states $\rho_0$ and $\rho_1$, respectively, there exists an explicit quantum algorithm that estimates $\Tinfty\ab(\rho_0, \rho_1)$ to within additive error $\epsilon$, using $O\ab\big(\log(1/\epsilon)/\epsilon^{3/2})$ queries to $Q_0$ and $Q_1$.
\end{theorem}
\begin{proof}
Given $(n+a)$-qubit unitary oracles $Q_0$ and $Q_1$ that prepare purifications $\ket{\psi_0}$ and $\ket{\psi_1}$ of the $n$-qubit quantum states $\rho_0$ and $\rho_1$, respectively. We can estimate $\Tinfty\rbra{\rho_0, \rho_1}$ as follows.

\begin{enumerate}
    \item For each $k \in \cbra{0, 1}$, 
\begin{enumerate}[label=1.\arabic*]
    \item Construct a unitary operator $U_{A_k}$ that is an $(n+2a+2)$-qubit unitary dilation of
    \[
    A_k \coloneq \frac{I - (-1)^{k}\Delta}{2}, \qquad \text{where } \Delta \coloneq \frac{\rho_0-\rho_1}{2},
    \]
    using $O(1)$ queries to $Q_0$ and $Q_1$ by \cref{lemma:dilation-Delta} and linear combination of unitaries~\cite{CW12}.
    \item Construct a qubitization unitary operator $W_{A_k} \coloneq U_{A_k}R(U_{A_k})^\dag R$ where $R = 2\Pi - I_{n+2a+2}$ and $\Pi = \ketbra{0}{0}^{\otimes (2a+2)} \otimes I_n$, using $O(1)$ queries to $U_{A_k}$. 
    \item Apply the maximum phase estimation algorithm in \cref{lem:max_phase_estimation}, let $\tilde{\theta}_{\max}^{(k)}$ be the output of $\mathsf{MaxPhaseEst}\ab*(-W_{A_k},\, I_{2a+2} \otimes Q_k,\, \epsilon,\, \gamma)$, where $\gamma = \sqrt{\epsilon}$.
    \item Let $\tilde{\lambda}_{\max}^{(k)} = \max\{0, 1 - 2\cos\ab*((\pi - \tilde{\theta}^{(k)}_{\max})/{2})\}$.
\end{enumerate}
    \item Return $\max\ab*\{\tilde{\lambda}_{\max}^{(0)}, \tilde{\lambda}_{\max}^{(1)}\}$.
\end{enumerate}
The above algorithm is formally given in \cref{alg:T_inf_estimation}.

\begin{algorithm}[ht]
\caption{Operator norm distance estimation}
\label{alg:T_inf_estimation}
\begin{algorithmic}[1]
\Require $(n+a)$-qubit unitary oracles $Q_0$ and $Q_1$ that prepare purifications of $n$-qubit quantum states $\rho_0$ and $\rho_1$, precision $\epsilon \in (0, 0.1)$.
\Ensure An estimate of $\Tinfty(\rho_0, \rho_1)$.
 
\State Denote $\Delta \coloneq (\rho_0-\rho_1)/2$ and $R \coloneq 2\Pi - I_{n+2a+2}$, where $\Pi \coloneq \ketbra{0}{0}^{\otimes(2a+2)} \otimes I_n$.

\For{$k \in \{0,1\}$}
    \State Construct an $(n+2a+2)$-qubit unitary dilation $U_{A_k}$ of
    $A_k \coloneq \rbra{I-(-1)^k\Delta}/{2}$.
    \State Construct a qubitization unitary $W_{A_k} \coloneq U_{A_k} R \rbra{U_{A_k}}^\dagger R$.
    \State $\tilde{\theta}^{(k)}_{\max} \gets \mathsf{MaxPhaseEst}\ab*(-W_{A_k},\, I_{2a+2} \otimes Q_k,\, \epsilon,\, \gamma)$, where $\gamma = \sqrt{\epsilon}$.
    \State $\tilde{\lambda}_{\max}^{(k)} \gets \max\{0, 1 - 2\cos\ab*((\pi - \tilde{\theta}^{(k)}_{\max})/{2})\}$.
\EndFor

\State \Return $\max\{\tilde{\lambda}_{\max}^{(0)}, \tilde{\lambda}_{\max}^{(1)}\}$.
\end{algorithmic}
\end{algorithm}

First we analyze the complexity of \cref{alg:T_inf_estimation}. Note that the unitary dilation $U_{A_k}$ can be implemented using $O(1)$ queries to $Q_0$ and $Q_1$ by \cref{lemma:dilation-Delta} and linear combination of unitaries~\cite{CW12}. The qubitization unitary $-W_{A_k}$ uses $O(1)$ queries to $U_{A_k}$. Taking $\gamma = \sqrt{\epsilon}$, the subroutine $\mathsf{MaxPhaseEst}(-W_{A_k}, I_{2a+2}\otimes Q_k, \epsilon, \gamma)$ uses a total number of $O(\log(1/\epsilon)/\epsilon^{3/2})$ queries to $Q_0$ and $Q_1$. Therefore, the total number of queries to $Q_0$ and $Q_1$ is $O(\log(1/\epsilon)/\epsilon^{3/2})$.

Next, we prove the correctness of the algorithm. We now consider the properties of $\tilde{\theta}^{(k)}_{\max}$ for each $k \in \{0, 1\}$. 
For convenience, here we assume that $\lambda_{\max}\rbra{(-1)^k\Delta} \geq \epsilon$ (and the corner case where $\lambda_{\max}\rbra{(-1)^k\Delta} < \epsilon$ will be discussed later).
Note that the random variable $\tilde{\theta}^{(k)}_{\max}$ is the output of $ \mathsf{MaxPhaseEst}\ab*(-W_{A_k},\, I_{2a+2} \otimes Q_k,\, \epsilon,\, \gamma)$, where $\gamma = \sqrt{\epsilon}$. 
Then, it holds (shown later) that
\begin{equation}\label{eq:subroutine_ouput_mixed}
    \Pr\sbra[\Big]{\abs[\big]{\tilde{\theta}^{(k)}_{\max} - \theta_{\max}^{(k)} }\leq \epsilon} \geq 0.99, \qquad \text{where } \theta_{\max}^{(k)} = \pi + 2\arccos\ab\Big(\frac{1-\lambda_{\max}((-1)^k\Delta)}{2}).
\end{equation}
Therefore, when we take
\[
\tilde{\lambda}_{\max}^{(k)} = \max\ab\bigg\{0,  1 - 2\cos\ab\Big(\frac{\pi - \tilde{\theta}_{\max}^{(k)}}{2})\},
\]
and return $\tilde\lambda_{\max} = \max\{\tilde{\lambda}_{\max}^{(0)}, \tilde{\lambda}_{\max}^{(1)}\}$, it holds that
\begin{equation}\label{eq:error_two_mixed}
    \Pr\sbra*{\abs[\big]{\tilde{\lambda}_{\max} - \Tinfty(\rho_0, \rho_1)} \leq \epsilon} \geq 0.99^2 > 0.98.
\end{equation}
To see this, note that $\Tinfty(\rho_0, \rho_1) = \max\{\lambda_{\max}(\Delta), \lambda_{\max}(-\Delta)\}$ by definition. Therefore, by the mean value theorem, there exists some $\xi^{(k)}$ between $\theta_{\max}^{(k)}$ and $\tilde{\theta}_{\max}^{(k)}$ such that
\begin{align*}
    \abs{\tilde{\lambda}_{\max}^{(k)} - \lambda_{\max}((-1)^k\Delta)} &= \abs[\Big]{\sin\ab\Big(\frac{\xi^{(k)}-\pi}{2})} \abs{\tilde{\theta}_{\max}^{(k)} - \theta_{\max}^{(k)}}.
\end{align*}
Since $\abs{\sin\ab((\xi-\pi)/2)} \leq 1$, we have
\[
\Pr\sbra*{\abs[\big]{\tilde{\lambda}_{\max}^{(k)} - \lambda_{\max}((-1)^k\Delta) }\leq \epsilon} \geq 0.99,
\]
which gives \cref{eq:error_two_mixed}. Therefore, \cref{alg:T_inf_estimation} returns an estimate of $\Tinfty(\rho_0, \rho_1)$ to within additive error $\epsilon$.

To complete the proof, it remains to prove that the output of two subroutines satisfy the condition in \cref{eq:subroutine_ouput_mixed}. We begin with the analysis on the subroutine $\mathsf{MaxPhaseEst}(-W_{A_0}, I_{2a+2} \otimes Q_0, \epsilon, \gamma)$. To show that applying the quantum maximum phase estimation algorithm in \cref{lem:max_phase_estimation} gives a result satisfying \cref{eq:subroutine_ouput_mixed}, we prove the following properties in order.
\begin{enumerate}
    \item The eigenphases of $-W_{A_0}$ are in $[0, 2\pi - 2\epsilon)$, and the largest eigenphase of $-W_{A_0}$ is $\theta_{\max}^{(0)} = \pi + 2\arccos\ab((1-\lambda_{\max}(\Delta))/{2})$.
    \item The state preparation unitary $I_{2a+2} \otimes Q_0$ prepares a state that has an overlap at least $\gamma = \sqrt{\epsilon}$ with the top eigenspace of $-W_{A_0}$.
\end{enumerate}

Let $d=2^n$ be the dimension of the quantum states $\rho_0$ and $\rho_1$. For $j\in \{1,\dots, d\}$, let $\ket{v_j}$ denote an eigenvector of $\Delta$ with eigenvalue $\lambda_j$, i.e., $\Delta\ket{v_j} = \lambda_j \ket{v_j}$, then $\ket{v_j}$ is also an eigenvector of $A_0$ such that
\[
A_0 \ket{v_j} = \mu_j \ket{v_j}, \qquad \mu_j \coloneq \frac{1-\lambda_j}{2}.
\]
Note that the eigenvalues of $\Delta$ are in $[-1/2, 1/2]$, and therefore the eigenvalues of $A_0$ are in $[1/4, 3/4]$.

Let $\ket{u_j} = \ket{0}^{\otimes \rbra{2a+2}}\otimes\ket{v_j}$, and define 
\[
    \ket{u_j^\perp} \coloneq \frac{U_{A_0}\ket{u_j} - \mu_j\ket{u_j}}{\sqrt{1 - \mu_j^2}}.
\]
By \cref{lemma:qubitization}, since $\abs{\mu_j}<1$, the qubitization unitary $W_{A_0}$ has eigenvectors 
\begin{align*}
        \ket{\phi^{\pm}_{j}} & = \frac{1}{\sqrt{2}} \rbra[\big]{ \ket{u_j} \mp i \ket{u_j^\perp}}, 
\end{align*}
with $W_{A_0} \ket{\phi^{\pm}_{j}} = e^{\pm i \vartheta_j} \ket{\phi^{\pm}_{j}}$, where
\[
\vartheta_j \coloneq 2\arccos(\mu_j) = 2\arccos\ab\Big(\frac{1-\lambda_j}{2}).
\]
Moreover, let $\mathcal{P} \coloneqq \spanspace\set{\ket{u_j}}{1 \leq j \leq 2^n}$, $\mathcal{Q} \coloneq \spanspace\set{\ket{u_j^\perp}}{1 \leq j \leq 2^n}$, and $\mathcal{M} = \rbra{\mathcal{P} + \mathcal{Q}}^{\perp}$, then every nonzero vector $\ket{\phi'} \in \mathcal{M}$ is an eigenvector of $W_{A_0}$ with eigenphase $0$, i.e., $W_{A_0} \ket{\phi'} = \ket{\phi'}$.

Since the function $\arccos((1-x)/{2})$ is strictly monotonically increasing on $[-1/2, 1/2]$, for increasing eigenvalues of $\Delta$, 
\[
\lambda_1 \leq \lambda_2 \leq \cdots \leq \lambda_{d},
\]
we have
\[
\vartheta_1 \leq \vartheta_2 \leq \cdots \leq \vartheta_{d}, \quad\text{where }
\vartheta_{j} \in \ab[2\arccos\ab(3/4), 2\arccos\ab(1/4)] \subset (0, \pi-2\epsilon).
\]
Now we add a global phase $e^{i\pi}$ to $W_{A_0}$, i.e., $e^{i\pi}W_{A_0} = - W_{A_0}$. The phases of $- W_{A_0}$, denoted by $\theta^{\pm}_{j} = \pi \pm \vartheta_j$, lie in $[0, 2\pi-2\epsilon)$ as $\pi - \vartheta_j \in (0,\pi-2\epsilon)$ and $\pi + \vartheta_j \in (\pi,2\pi-2\epsilon)$. The largest phase of $- W_{A_0}$, denoted by $\theta_{\max}^{(0)} \coloneq \max_j\{\theta^{\pm}_{j}\} = \max_j\{\theta^{+}_{j}\}$, clearly corresponds to the largest eigenvalue of $\Delta$, 
\begin{equation}\label{eq:max_eigenphase}
    \theta_{\max}^{(0)} = \pi+\vartheta_{N} = \pi + 2\arccos\ab(\frac{1-\lambda_{\max}(\Delta)}{2}).
\end{equation}

Then we show that the state preparation unitary $I_{2a+2} \otimes Q_0$ prepares a state that has an overlap at least $\gamma = 1/2$ with the top eigenspace of $-W_{A_0}$. The unitary $I_{2a+2} \otimes Q_0$ prepares the input state
\[
\ket{\Psi_0} \coloneq (I_{2a+2} \otimes Q_0) \ket{0}^{\otimes \rbra{n+3a+2}} = \ket{0}^{\otimes \rbra{2a+2}} \ket{\psi_0}
\]
where $\ket{\psi_0}$ is the purification of the state $\rho_0$. By \cref{lem:top_eigenspace_overlap}, the overlap between the input state and the top eigenspace of $-W_{A_0}$ is at least
\begin{align*}
    \sum_{j:\theta^{+}_{j} = \theta_{\max}^{(0)}} \norm[\big]{(\bra{\phi_{j}^{+}}\otimes I_a) \ket{\Psi_0}}_2^2 &= \sum_{j:\theta^{+}_{j} = \theta_{\max}^{(0)}} \norm[\Big]{\ab\Big(\frac{1}{\sqrt{2}} \ab\big( \bra{u_j} + i \bra{u_j^\perp})\otimes I_a) \ket{\Psi_0}}_2^2 \\
    &\geq \sum_{j:\theta^{+}_{j} = \theta_{\max}^{(0)}} \norm[\Big]{\frac{1}{\sqrt{2}}\ab\big(\bra{u_j} \otimes I_a) \ket{\Psi_0}}_2^2\\
    &= \sum_{j:\lambda_{j}=\lambda_{\max}(\Delta)} \norm[\Big]{\frac{1}{\sqrt{2}}\big( \bra{0}^{\otimes {2a+2}} \otimes \bra{v_j}  \otimes I_a)\ab\big(\ket{0}^{\otimes \rbra{2a+2}} \otimes \ket{\psi_0})}_2^2\\
    &= \sum_{j:\lambda_{j}=\lambda_{\max}(\Delta)} \norm[\Big]{\frac{1}{\sqrt{2}}\big( \bra{v_j}  \otimes I_a)\ket{\psi_0}}_2^2\\
    &= \frac{1}{2} \norm[\big]{\ab\big( \Pi_{\max}^{(\Delta)}\otimes I_{a})\ket{\psi_0}}_2^2\\
    &\ge \lambda_{\max}(\Delta).
\end{align*}
With $\lambda_{\max}(\Delta) \geq \epsilon$, we can take $\gamma = \sqrt{\epsilon}$. Therefore, by \cref{lem:max_phase_estimation}, the maximum phase estimation algorithm $\mathsf{MaxPhaseEst}(-W_{A_0}, I_{2a+2} \otimes Q_0, \epsilon, \gamma)$ uses $O(1/\sqrt{\epsilon})$ queries to $Q_0$ and $O(\log(1/\epsilon)/\epsilon^{3/2})$ queries to $-W_{A_0}$, outputs an estimate $\tilde{\theta}_{\max}^{(0)}$ such that
\[
\Pr\sbra[\Big]{\abs[\big]{\tilde{\theta}^{(0)}_{\max} - \theta_{\max}^{(0)} }\leq \epsilon} \geq 0.99,
\]
as in \cref{eq:subroutine_ouput_mixed} for $k=0$.

Note that if $\lambda_{\max}(\Delta) \in [0, \epsilon)$, taking $\gamma = \sqrt{\epsilon}$ may not satisfy the overlap promise since $\gamma^2 = \epsilon > \lambda_{\max}(\Delta)$. However, even without the overlap promise, the maximum phase estimation algorithm $\mathsf{MaxPhaseEst}(-W_{A_0}, I_{2a+2} \otimes Q_0, \epsilon, \gamma)$ still outputs an estimate $\tilde{\theta}_{\max}^{(0)}$, and with probability at least $0.99$, there exists an eigenphase $\theta_j^{\pm}$ of $-W_{A_0}$ such that $\abs{\tilde{\theta}_{\max}^{(0)} - \theta_j^{\pm}} \leq \epsilon$. Let $f(\theta) = 1 - 2\cos((\pi - \theta)/2)$, by the mean value theorem, there exists a point $\xi^{(0)}$ between $\tilde{\theta}_{\max}^{(0)}$ and $\theta_j^{\pm}$ such that
\begin{align*}
    \abs{f(\tilde{\theta}_{\max}^{(0)}) - f(\theta_j^{\pm})} &= \abs{f'(\xi^{(0)})} \abs{\tilde{\theta}_{\max}^{(0)} - \theta_j^{\pm}}\\
    &= \abs[\Big]{\sin\ab\Big(\frac{\xi^{(0)}-\pi}{2})} \abs{\tilde{\theta}_{\max}^{(0)} - \theta_j^{\pm}}\\
    &\leq \abs{\tilde{\theta}_{\max}^{(0)} - \theta_j^{\pm}}\\
    &\leq \epsilon.
\end{align*}
By definition, $f(\theta_j^{\pm}) = \lambda_{j}$, thus we have
\[
f(\tilde{\theta}_{\max}^{(0)}) \leq \lambda_j + \epsilon \leq \lambda_{\max}(\Delta) + \epsilon.
\]
Since $0 \leq \lambda_{\max}(\Delta) < \epsilon$, taking $\tilde{\lambda}_{\max}^{(0)} = \max\{0, f(\tilde{\theta}_{\max}^{(0)})\}$ still guarantees that
\[
\Pr\sbra[\Big]{\abs[\big]{\tilde{\lambda}^{(0)}_{\max} - \lambda_{\max}(\Delta) }\leq \epsilon} \geq 0.99.
\]

Similarly, we can analyze the subroutine $\mathsf{MaxPhaseEst}(-W_{A_1}, I_{2a+2}\otimes Q_1, \epsilon, \gamma)$ by replacing $\Delta$ with $-\Delta$. Note that we have $\norm{\ab*(\Pi_{\max}^{(-\Delta)}\otimes I_{a})\ket{\psi_1}}_2^2 \ge 2\lambda_{\max}(-\Delta)$ by symmetry. We can use the same analysis above to show that the algorithm $\mathsf{MaxPhaseEst}(-W_{A_1}, I_{2a+2} \otimes Q_1, \epsilon, \gamma)$ uses $O(1/\sqrt{\epsilon})$ queries to $Q_1$ and $O(\log(1/\epsilon)/\epsilon^{3/2})$ queries to $-W_{A_1}$, outputs an estimate $\tilde{\theta}_{\max}^{(1)}$ such that
\[
\Pr\sbra[\Big]{\abs[\big]{\tilde{\theta}^{(1)}_{\max} - \theta_{\max}^{(1)} }\leq \epsilon} \geq 0.99,
\]
as in \cref{eq:subroutine_ouput_mixed} for $k=1$.
\end{proof}

\subsection{Quantum sample algorithm for estimating operator norm  distance}
\label{subsec:general-sample}

\begin{theorem} \label{thm:sample-general}
    There is a quantum algorithm that estimates  $\Tinfty\rbra{\rho_0, \rho_1}$ to within additive error $\epsilon$, using $O\rbra{\log^2(1/\epsilon)/\epsilon^{3}}$ samples of $\rho_0$ and $\rho_1$. 
\end{theorem}

\begin{proof}
    This is obtained immediately by \cref{lemma:lifting} using the query complexity $O\rbra{\log(1/\epsilon)/\epsilon^{3/2}}$ given in \cref{thm:query-general}, which gives the sample complexity $O\rbra{\log^2(1/\epsilon)/\epsilon^{3}}$. 
\end{proof}

\section*{Acknowledgments}
\noindent The work of Yupan Liu was supported by funding from the Swiss State Secretariat for Education, Research and Innovation (SERI). 
The work of Qisheng Wang was supported by startup funding from Shanghai Jiao Tong University. 
The work of Zhan Yu was supported by the CQT Young Researcher Career Development Grant.
ChatGPT was used interactively to explore possible approaches to the research problems addressed in main technical sections, identify relevant references, and proofread the manuscript, while all writing, including mathematical statements and reasoning, was completed by the authors.

\bibliographystyle{alphaurl}
\bibliography{Linfty}

\appendix

\end{document}